\documentclass[a4, twocolumn, 10pt]{IEEEtran}

\usepackage{amsmath,epsfig,amssymb,verbatim,amsopn,cite,subfigure,multirow}
\usepackage{balance,nopageno}
\usepackage{algorithm,algorithmic}
\usepackage[usenames,dvipsnames]{color}
\usepackage[all]{xy}  
\usepackage{url}
\usepackage{amsfonts}
\usepackage{amssymb}
\usepackage{epsfig}
\usepackage{epstopdf}
\usepackage{slashbox}
\usepackage{bm}
\usepackage{cases}
\usepackage[margin=19.5mm,top=18mm,bottom=18mm]{geometry}

\newtheorem{Remark}{Remark}

  {\proof}{\proofend}
\newtheorem{proposition}{Proposition}

\DeclareMathOperator*{\argmax}{arg\,max}

\newcommand{\ve}[1]{\boldsymbol{#1}}

\newcommand{\vH}{\ve{H}} \newcommand{\vh}{\ve{h}}

 \newcommand{\vw}{\ve{w}}

\newcommand{\be}{\begin{equation}} \newcommand{\ee}{\end{equation}}
\newcommand{\bea}{\begin{eqnarray}} \newcommand{\eea}{\end{eqnarray}}

\newcommand{\qa}{{\bf a}}

\newcommand{\qh}{{\bf h}}

\newcommand{\qn}{{\bf n}}

\newcommand{\qr}{{\bf r}}

\newcommand{\qu}{{\bf u}}

\newcommand{\qw}{{\bf w}}
\newcommand{\qx}{{\bf x}}

\newcommand{\qA}{{\bf A}}
\newcommand{\qB}{{\bf B}}

\newcommand{\qD}{{\bf D}}

\newcommand{\qH}{{\bf H}}
\newcommand{\qI}{{\bf I}}

\newcommand{\qR}{{\bf R}}

\newcommand{\qW}{{\bf W}}

\newcommand{\qY}{{\bf Y}}
\newcommand{\qZ}{{\bf Z}}

\newcommand{\hRD}{\vh_{RD}}
\newcommand{\hSR}{\vh_{SR}}

\newcommand{\thSR}{\tilde{\vh}_{SR}}
\newcommand{\thRD}{\tilde{\vh}_{RD}}

\newcommand{\MRT}{\mathsf{MRT}}
\newcommand{\MRC}{\mathsf{MRC}}
\newcommand{\TZF}{\mathsf{TZF}}
\newcommand{\RZF}{\mathsf{RZF}}
\newcommand{\ZF}{\mathsf{ZF}}

\newcommand{\mSINR}{\mathsf{mSINR}}
\newcommand{\gamth}{\gamma_{\mathsf{th}}}
\newcommand{\Pout}{\mathsf{P_{out}}}

\newcommand{\WT}{\vw_t}

\newcommand{\diag}{\mathsf{diag}}
\newcommand{\Bta}{\mathsf{Beta}}

\newcommand{\SnR}{\sigma_R^2}
\newcommand{\SnD}{\sigma_{D}^2}

\newcommand{\Srd}{d_2^{\tau}}
\newcommand{\Ssr}{d_1^{\tau}}

\newcommand{\bGsd}{ \frac{d_1^{\tau} z}{ \rho_1}}

\newcommand{\Sap}{\sigma_{RR}^2}

\newcommand{\Prob}{\textnormal{Pr}}

\usepackage{multibib}
\usepackage[nodisplayskipstretch]{setspace}
\newcites{Prim}{Very important papers}

\definecolor{light-gray}{gray}{0.65}

\newcounter{mytempeqcounter}

\newcommand{\bigformulatop}[2]{%
  \begin{figure*}[!t]
    \normalsize
    \setcounter{mytempeqcounter}{\value{equation}}
    \setcounter{equation}{#1}
    #2

    \setcounter{equation}{\value{mytempeqcounter}}
    \hrulefill
    \vspace*{4pt}
  \end{figure*}
}

\title{Throughput Analysis and Optimization of Wireless-Powered Multiple Antenna Full-Duplex Relay Systems}
\author{\normalsize {Mohammadali Mohammadi,~\IEEEmembership{Member,~IEEE,}
Batu K. Chalise,~\IEEEmembership{Senior Member,~IEEE,}\\
 Himal A. Suraweera,~\IEEEmembership{Senior Member,~IEEE,}
 Caijun Zhong,~\IEEEmembership{Senior Member,~IEEE,}\\
 \hspace{2.5em}Gan Zheng,~\IEEEmembership{Senior Member,~IEEE,}
 and Ioannis Krikidis,~\IEEEmembership{Senior Member,~IEEE}}
  \thanks{
    Mohammadali Mohammadi is with the Faculty of Engineering, Shahrekord University, Shahrekord 115, Iran
    (email: m.a.mohammadi@eng.sku.ac.ir).}
    \thanks{
    Batu K. Chalise is with Cleveland State University, 2121 Euclid Avenue, Cleveland, OH 44115 (email: b.chalise@csuohio.edu). }
  \thanks{
    Himal A. Suraweera is with the Department of Electrical and Electronic
Engineering, University of Peradeniya, Peradeniya 20400, Sri Lanka  (email:
    himal@ee.pdn.ac.lk). }
      \thanks{
    Caijun Zhong is with the Department of Information Science and Electronic
Engineering, Zhejiang University, Hangzhou 310027, China (email:  caijunzhong@zju.edu.cn). }
      \thanks{
   Gan Zheng is with School of Computer Science and Electronic Engineering, University of Essex, UK (email: ganzheng@essex.ac.uk). }
    \thanks{I. Krikidis is with the Department of Electrical and Computer
Engineering, University of Cyprus, Nicosia 1678, Cyprus (email:
    krikidis@ucy.ac.cy).}
      \thanks{
This work was presented in part at the IEEE International Workshop on Signal Processing Advances in Wireless Communications (SPAWC 2015), Stockholm, Sweden, June/July 2015.
}
}

\begin{document}

\maketitle
\thispagestyle{empty}

\begin{abstract}
We consider a full-duplex (FD) decode-and-forward system in which the time-switching protocol is employed by the multi-antenna relay to receive energy from the source and transmit information to the destination. The instantaneous throughput is maximized by optimizing receive and transmit beamformers at the relay and the time-split parameter. We study both optimum and suboptimum schemes. The reformulated problem in the optimum scheme achieves closed-form solutions in terms of transmit beamformer for some scenarios. In other scenarios, the optimization problem is formulated as a semi-definite relaxation problem and a rank-one optimum solution is always guaranteed. In the suboptimum schemes, the beamformers are obtained using maximum ratio combining, zero-forcing, and maximum ratio transmission. When beamformers have closed-form solutions, the achievable instantaneous and delay-constrained throughput are analytically characterized. Our results reveal that, beamforming increases both the energy harvesting and loop interference suppression capabilities at the FD relay. Moreover, simulation results demonstrate that the choice of the linear processing scheme as well as the time-split plays a critical
role in determining the FD gains.
\end{abstract}

\section{Introduction}\label{GENERAL_Introduction}
The emergence of multimedia rich wireless services coupled with ever growing number of subscribers has placed a high demand for radio resources such as bandwidth and energy. Most wireless radios so far have adopted half-duplex (HD) operation where uplink and downlink communication are ``orthogonalized'' in either time or frequency domain, which leads to a loss of spectral efficiency. An attractive solution to improve the spectral efficiency is to allow full-duplex (FD) simultaneous transmission/reception at the expense of loopback interference (LI) caused by the signal leakage from the transceiver output to the input~\cite{Sabharwal:JSac:2014,Margetts:JSac:2014,Riihonen:TSP:2011,Riihonen:TWC:2011,Duarte:Thesis:2012}.

Traditionally, LI suppression has been performed using passive isolation techniques such as placing RF absorber material between antennas, deploying of directional antennas.~\cite{Korpi:TWC:2014}. These schemes alone are inadequate to suppress the LI below the noise floor level required in most wireless systems. To achieve more effective suppression, a FD node could apply time-domain active techniques to pre-cancel the radio frequency (RF) (analog domain) or baseband (digital domain) LI signal~\cite{Duarte:Thesis:2012,Aryafar:2012}. However, such mitigation schemes require sophisticated electronic implementation~\cite{Riihonen:TSP:2011}. With the ubiquitous use of multi-antenna wireless systems, spatial domain precoding techniques can also be deployed at MIMO FD nodes. Such techniques have received significant interest as an attractive FD solution~\cite{Suraweera:TWC:2014,Ngo:JSAC:2014}.

On the other hand, many contemporary communication systems are battery powered and have a limited operational lifetime. To this end energy harvesting communications is a new paradigm that can power wireless devices by scavenging energy from external resources such as solar, wind, ambient RF power etc.~\cite{Mehta:TWC:2010}. However, energy harvesting from such sources are not without challenges due to the unpredictable nature of these energy sources. To this end, wireless energy transfer has been touted as a promising technique for a variety of wireless applications~\cite{Lumpkins:Mag:2014,Huang:COMM:2015,Rui_Zhang:TWC:2013}.

RF signals can carry both information and energy and this fundamental tradeoff has been studied in~\cite{Varshney:ISIT:2008,Sahai:ISIT:2010}. In order to address practical issues associated with simultaneous information and energy transfer (same signal can not be used for both decoding and rectifying), two practical approaches, i.e., time-switching (TS) and power-splitting (PS) point-to-point system architectures were proposed in~\cite{Zhou:TCOM:2014}. Subsequent works have also considered wireless-powered HD relay transmission with TS and PS architectures, for example, different relay networks have been studied considering amplify-and-forward (AF) and decode-and-forward (DF) relaying~\cite{Nasir:TWC:2013,Nasir:TCOM:2015,YonghuiLi:TSP:2015}, large scale networks~\cite{Krikidis:TCOM:2014} and multiple antenna relay systems~\cite{Krikidis:TCOM:2015,Zhu:TCOM:2015}.

Inspired by the benefits of FD and wireless power transfer, some recent papers have investigated the performance of wireless-powered FD point-to-point~\cite{Yamazaki:WCNC:2015,Ju:TCOM:2014} and relay systems~\cite{Caijun:TCOM:2014,Zeng:WCL:2015}. In~\cite{Ju:TCOM:2014} a wireless network model with a hybrid FD access-point (AP) that broadcasts wireless energy to a set of downlink users and at the same time receives information from the users in the uplink has been considered. In~\cite{Yamazaki:WCNC:2015} software-defined radio implementation of a wireless system that transmits data and power in the same frequency has been presented. In~\cite{Caijun:TCOM:2014}, the achievable throughput of FD AF and DF relaying systems with TS has been studied. In~\cite{Zeng:WCL:2015} the performance of a wireless-powered AF relaying system has been also studied. The protocol in~\cite{Zeng:WCL:2015} considers energy harvesting from LI and therefore can recycle some part of the relay transmit energy. However,~\cite{Caijun:TCOM:2014} and~\cite{Zeng:WCL:2015} only assumed single transmit/receive antennas at the relay.

Inspired by the current work on wireless-powered FD, in this paper we consider a two-hop MIMO relay system where the multiple antenna FD relay is powered via wireless energy transfer from the source. The main motivation for the adoption of multiple antennas at the relay is two-fold: (1) employment of an antenna array helps the relay to accumulate more energy (2) spatial LI cancellation techniques can be deployed. Specifically, we design receive and transmit beamformers at the relay and optimize TS parameter to characterize instantaneous as well as delay-constrained throughputs.  The main contributions of the paper are summarized as follows:
\begin{enumerate}
\item Optimum as well as suboptimum schemes for maximizing the instantaneous throughput are proposed.  In the optimum scheme, the beamformer optimization is reformulated in terms of the transmit beamformer which is shown to have closed-form solutions for some scenarios. In other scenarios where such solutions are not available, we reformulate the optimization problem as a semi-definite relaxation (SDR) problem in terms of
a transmit beamforming matrix. The resulting optimization can be solved as a convex feasibility problem. We prove that the SDR problem either yields optimum rank-one solution or such solution can always be recovered from the  optimum beamforming matrix solution of the SDR problem.

\item In the suboptimum schemes,  we employ zero-forcing (ZF), maximum ratio combining (MRC), and maximum ratio transmission (MRT) schemes for obtaining receive and transmit beamfomers. More specifically, we solve the optimization problems that arise due to the application of  transmit zero-forcing (TZF)/ MRC and MRT/receive zero-forcing (RZF) as transmit/receive beamformers.

 \item In all of the above schemes, the optimum time-split parameter is analytically determined.

\item For the suboptimum schemes, which yield closed-form solutions, we develop new expressions for the system's outage probability. These expressions are helpful for investigating the effects of key system parameters on performance metrics such as the outage probability and delay-constrained throughput.

\item We present simple high signal-to-noise ratio (SNR) expressions for the outage probability of suboptimum schemes which enable the characterization of the system's diversity order and array gain.
\end{enumerate}

The remainder of the paper is organized as follows: Section~\ref{GENERAL_system_model} presents the multiple antenna FD relay system model. Section~\ref{sec:Precoder designs} introduces joint transmit/receive beamforming designs. The instantaneous and delay-constrained throughput of these beamforming schemes are analyzed in Section~\ref{sec:IT} and~\ref{sec:DCT}, respectively. Numerical results are reported in Section~\ref{sec:numerical results}. Finally, Section~\ref{sec:conclusion} concludes the paper and summarizes the key findings.

\emph{Notation:} We use bold upper case letters to denote matrices, bold lower case letters to denote vectors. $\|\cdot\|$, $(\cdot)^{\dag}$, $(\cdot)^{-1}$ and $\mathsf{tr} (\cdot)$  denote the Euclidean norm, conjugate transpose operator, matrix inverse and the trace of a matrix respectively; ${\tt E}\left\{x\right\}$ stands for the expectation of the random variable $x$; $\Prob(\cdot)$ denotes the probability; $f_X(\cdot)$ and $F_X(\cdot)$ denote the probability density function (pdf) and cumulative distribution function (cdf) of the random variable (RV) $X$, respectively; $\mathcal{CN}(\mu,\sigma^2)$ denotes a circularly symmetric complex Gaussian RV $x$ with mean $\mu$ and variance $\sigma^2$;  $\Gamma(a)$ is the Gamma function; $\Gamma(a,x)$ is upper incomplete Gamma function~\cite[Eq. (8.350)]{Integral:Series:Ryzhik:1992}; $K_{\nu}(\cdot)$ is the $\nu$th order modified Bessel function of the second kind~\cite[Eq. (8.432)]{Integral:Series:Ryzhik:1992}; $\psi(x)$ is the digamma function~\cite[Eq. (6.3.1)]{Abramowitz_Handbook_1970}; $E_n(x)$ is the E$_{n}$-function~\cite[Eq. (5.1.4)]{Abramowitz_Handbook_1970} and \small{$G_{p q}^{m n} \left( z \  \vert \  {a_1\cdots a_p \atop b_1\cdots b_q} \right)$ }\normalsize
denotes the Meijer G-function~\cite[ Eq. (9.301)]{Integral:Series:Ryzhik:1992}.
\section{System Model and Problem Formulation}\label{GENERAL_system_model}
We consider a multiple antenna DF relay system consisting of one
source $S$, one relay $R$, and one destination, $D$ as shown in Fig.~\ref{fig: system model}. Both $S$ and $D$ are equipped with a single antenna. To enable FD operation, $R$ is equipped with two sets of antennas, i.e., $M_R$ receiving antennas and $M_T$ transmitting antennas. We assume that the $S$ to $D$ link does not exist.

It is also assumed that $R$ has no external power supply, and is powered through wireless energy transfer from $S$ as in~\cite{Nasir:TWC:2013,Caijun:TCOM:2014}. We adopt the TS protocol~\cite{Zhou:TCOM:2014,Nasir:TWC:2013}, hence the entire communication process is divided into two phases, i.e., for a transmission block time $T$,
$\alpha$ fraction of the block time is devoted for energy harvesting and the remaining time, $(1-\alpha)T$, is used for information transmission. It is also assumed that the channels experience Rayleigh fading and remain constant over the block time $T$ and varies independently and identically from one block to the other.

During the energy harvesting phase, the received signal $\qr_e$ at the relay can be expressed as
\begin{align}
  \qr_e= \sqrt{\frac{P_S}{d_1^{\tau}}} \qh_{SR}x_e + \qn_R,
\end{align}
where $P_S$ is the source transmit power, $d_1$ is the distance between
the source and relay, $\tau$ is the path loss exponent, $\qh_{SR}$ is the $M_R \times 1$ channel vector for the $S$-$R$ link, i.e., input antennas at $R$ are connected to the rectennas,\footnote{Another design choice would be to use all receive and transmit antennas during the energy harvesting phase. Analyzing this case involves different RVs in SINR expressions and left out as future work due to limited space. Moreover, as observed in single/dual antenna relay implementation~\cite{Caijun:TCOM:2014}, ``receive antennas'' versus ``all antennas'' design options are expected to follow similar trends reported in this work and achieve comparable performance depending on the operating SNR regime.} $x_e$ is the energy symbol with unit energy, and $\qn_R$ is the additive white Gaussian
noise (AWGN) at the relay with ${\tt E}\left\{\qn_R \qn_R^{\dag}\right\}=\sigma_{R}^2 \qI_{M_R}$.  As in \cite{Nasir:TWC:2013}, we assume that the energy harvested during the energy harvesting phase is stored in a supercapacitor and then fully consumed by $R$ to forward the source signal to the destination. This type of operation is also known as the ``harvest-use'' architecture in the literature~\cite{YonghuiLi:TSP:2015,Caijun:TCOM:2014}. We assume that the harvested energy due to the noise (including both the antenna noise and the rectifier noise) is small and thus ignored~\cite{Nasir:TWC:2013, Caijun:TCOM:2014,Zhou:2013}. Hence, the relay transmit power can be written as
\be\label{eqn:Pr}
P_r = \frac{\kappa}{d_1^{\tau}} P_S \|\qh_{SR}\|^2,
\ee
where $\kappa\triangleq\frac{\eta\alpha }{1-\alpha}$ and $\eta$ denotes the energy conversion efficiency.
Now, let us consider the information transmission phase. The received signal at $R$ can be expressed as
  \be\label{eqn:rn}
    \qr[n] = \sqrt{\frac{P_S}{d_1^{\tau}}} \qh_{SR}x_S[n] + \qH_{RR} \qx_R[n] + \qn_R[n],
  \ee
where $x_S[n]$ is the source information symbol with unit energy, and $\qx_R[n]$ is the transmitted relay signal satisfying ${\tt E}\left\{\qx_R[n]\qx^{\dag}_R[n]\right\}=P_r$.  In order to reduce the effects of LI on system performance, an imperfect interference cancellation (i.e. analog/digital cancellation) scheme\footnote{Perfect cancellation of LI is not possible due to imperfect estimation of LI channel, inevitable transceiver chain impairments~\cite{Riihonen:TSP:2011,Cirik:2014}, and inherent relay processing delay. Therefore, $\qH_{RR}$ can assume decently high values and its effect can be minimized with spatial suppression techniques. The training-based approach for estimating LI channel can be readily extended to retrieve $\qH_{RR}$ with a reasonably good accuracy. However, simulation results of this paper assume that the estimated and actual channels (including $\qH_{RR}$) are same. As such, the reported results serve as useful theoretical bounds for practical design.} is used at $R$ and we model the $M_R \times M_T$ residual LI channel $\qH_{RR}$ as a fading feedback channel. To this end, several residual LI channel models have been proposed in the literature, see for e.g.,~\cite{Riihonen:TWC:2011,Riihonen:TSP:2011,Duarte:Thesis:2012,Bliss:2012,Cirik:2014}. Since each implementation of a particular analog/digital LI cancellation scheme can be characterized by a specific residual power, the elements of $\qH_{RR}$ can be modeled as independent identically distributed $\mathcal{CN}(0,\Sap)$ RVs, which is a common assumption in the literature since the dominant line-of-sight component in LI can be removed effectively when a cancellation method is implemented~\cite{Korpi:TWC:2014,Duarte:Thesis:2012}.
\begin{figure}[t]
\centering
\includegraphics[width=85mm, height=50mm]{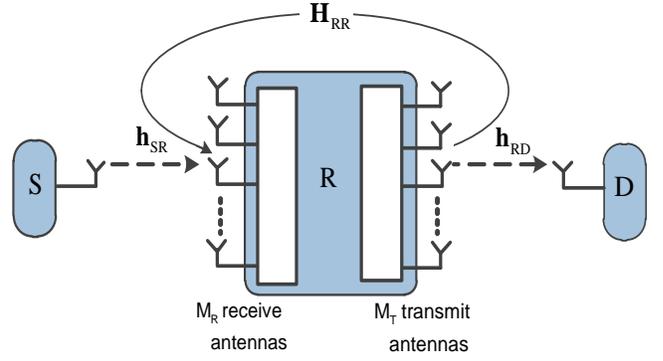}
\vspace{-0.4em}
\caption{Full-duplex relay system model.}
\vspace{-1.0em}
\label{fig: system model}
\vspace{0em}
\end{figure}
Since $R$ adopts the DF protocol, upon receiving the signal, it first applies a linear combining vector $\qw_r$ on $\qr[n]$ to obtain an estimate of $x_S$, then forwards the signal to $D$ using the transmit beamforming vector $\qw_t$. It is assumed that $\|\qw_t\|=\|\qw_r\|=1$.

The relay's estimate $\hat x_S[n]=\qw_r^{\dag} \qr[n]$ can be expressed as
  \begin{align}
   \hat x_S[n] &\!=\!  \sqrt{\frac{P_S}{\Ssr}} \qw_r^{\dag}\qh_{SR}x_S[n] \!+\! \qw_r^{\dag}\qH_{RR} \qx_R[n]\! +\!
   \qw_r^{\dag}\qn_R[n].
  \end{align}

The relay transmit signal is given by~\cite{Riihonen:TSP:2011}
\vspace{-0.1em}
    \be\label{eqn:xn}
    \qx_R[n] = \sqrt{P_r}\qw_t \hat x_S[n-\delta],
  \ee
where $\delta$ accounts for the time delay caused by relay processing. Finally, the received signal at $D$ is
expressed as
\vspace{-0.3em}
\be\label{eqn:ydn}
    y_D[n] = \sqrt{\frac{1}{\Srd}}\qh_{RD}\qx_R[n] + n_D[n].
  \ee
where $\qh_{RD}$ is the $1 \times M_T$ channel vector of the $R-D$ link, $d_2$ is the distance between $R$ and $D$ and $n_D$ denotes the AWGN at $R$ with ${\tt E}\left\{n_Dn_D^{\dag}\right\}=\sigma_{D}^2.$

With the DF protocol, end-to-end SINR can be written as
\vspace{-0.5cm}
\begin{align}
\gamma_{\mathsf{FD}} = &\min\left( \frac{\frac{\rho_1}{\Ssr}|\qw_r^{\dag} \qh_{SR}|^2}{ \frac{\kappa \rho_1}{\Ssr}
\|\qh_{SR}\|^2 |\qw_r^{\dag} \qH_{RR}\qw_t|^2  +
1},\right.\nonumber\\
&\qquad\qquad\qquad\left.\frac{\kappa \rho_2}{\Ssr\Srd} \|\qh_{SR}\|^2|
\qh_{RD}\qw_t|^2
    \right),\label{eq: e2e snr general}
\end{align}
where $\rho_1 = \frac{P_S}{\sigma_{R}^2}$ and $\rho_2 = \frac{P_S}{\sigma_{D}^2}$.

Using~\eqref{eq: e2e snr general}, the system's instantaneous rate is given by
\vspace{-0.4em}
\begin{align}
R(\alpha, \qw_t, \qw_r)=(1-\alpha)\log_2\left(1+ \gamma_{\mathsf{FD}} \right).\label{eq: e2e rate1}
\end{align}
Our objective is to maximize the rate $R(\alpha, \qw_t, \qw_r)$ w.r.t. $\alpha$, $\qw_t$, and $\qw_r$. This is mathematically expressed as
\vspace{-0.3em}
 \bea\label{eqn:optMain}
    \max_{\|\qw_r\|=\|\qw_t\|=1,~\!\alpha\in [0, 1)} && R(\alpha, \qw_t, \qw_r).
 \eea
The optimization problem \eqref{eqn:optMain} is  nonconvex and the challenge is to obtain optimum solutions efficiently. Towards this end, in the sequel, we propose  optimum and suboptimum schemes that solve \eqref{eqn:optMain}. More specifically, in the first step, we keep $\alpha$ fixed and propose optimum and suboptimum schemes for beamformer design. In the second step, $\alpha$ is optimized for the obtained beamformers. Although this two-step optimization problem requires a joint optimization w.r.t. $\alpha$ and $\qw_t$  in the optimum scheme, we show that it can still be solved efficiently.

\begin{Remark}
Apart from data-dependent transmit power, energy consumption in the power amplifier/RF circuitry and processing power of the LI cancellation technique could have a significant impact on FD transceiver operation~\cite{Auer:2011}. Towards this end, recent advances made with very low energy consumption micro-controllers and RF circuitry is already making it possible to use harvested power in realistic applications~\cite{Dostal:2015}. Nevertheless, performance analysis of wireless-powered FD transceiver operation with circuit and processing power consumption is an interesting future direction worth more research.
\end{Remark}
\vspace{-0.2cm}
\section{Joint Receive/Transmit Beamforming}\label{sec:Precoder designs}
In this section, we consider the beamforming design problem to solve the optimization problem \eqref{eqn:optMain} for a given $\alpha$. In this case, \eqref{eqn:optMain} turns to a problem of maximizing the minimum two-hop SINR, which is expressed as
\vspace{-0.2cm}
\begin{align}\label{eqn:opt}
    &\max_{ \|\qw_r\|=\|\qw_t\|=1}\min\left( \frac{  \frac{\rho_1}{\Ssr}|\qw_r^{\dag} \qh_{SR}|^2}{ \frac{\kappa \rho_1}{\Ssr}
\|\qh_{SR}\|^2 |\qw_r^{\dag} \qH_{RR}\qw_t|^2  + 1},
\right.\nonumber\\
&\qquad\qquad\qquad\qquad\qquad\left.\frac{\kappa\rho_2}{\Ssr\Srd} \|\qh_{SR}\|^2|
\qh_{RD}\qw_t|^2
    \right).
\end{align}
In the following subsections, we propose optimum as well as different suboptimum schemes for solving \eqref{eqn:opt}. In the optimum approach, we propose SDR problem but show that relaxation does not change the optimality of the solution, whereas in the suboptimum schemes, different linear  receiver/transmitter techniques are employed at $R$.

\vspace{-0.9em}
\subsection{Optimum Scheme}
Since the second-hop SNR does not depend on $\qw_r$, we can maximize the first-hop SINR w.r.t $\qw_r$ by fixing $\qw_t$. In this case, the optimization problem  \eqref{eqn:opt} is  re-formulated as
\vspace{-0.5em}
\begin{align}\label{eqn:opt:Rayleigh ratio problem}
   \max_{\|\qw_r\|=1} &  \frac{|\qw_r^{\dag} \qh_{SR}|^2}{ \frac{\kappa \rho_1}{\Ssr}
\|\qh_{SR}\|^2 |\qw_r^{\dag} \qH_{RR}\qw_t|^2+1},
\end{align}
which is a generalized Rayleigh ratio problem \cite{Horn}. It is well known that \eqref{eqn:opt:Rayleigh ratio problem} is globally maximized when 
\vspace{-0.2em}
 \be\label{eqn:opt2}
    \qw_r\! =\! \frac{\left( \frac{\kappa \rho_1}{\Ssr} \|\qh_{SR}\|^2
\qH_{RR}\qw_t\qw_t^\dag\qH_{RR}^\dag + \qI \right)^{-1}
\qh_{SR}}{\left\|\left(  \frac{\kappa \rho_1}{\Ssr} \|\qh_{SR}\|^2 \qH_{RR}\qw_t\qw_t^\dag\qH_{RR}^\dag + \qI \right)^{-1}
\qh_{SR}\right\|}.
 \ee
Accordingly, by substituting $\qw_r$ into \eqref{eqn:opt:Rayleigh ratio problem} and applying the Sherman Morrison formula \cite{Sherman}, the first term inside the $\min$ operator in  \eqref{eqn:opt} is obtained as
\vspace{-0.5em}
  \begin{align}
 \psi &\triangleq   \frac{\rho_1}{\Ssr}\left(\qh_{SR}^\dag\left(  \frac{ \kappa \rho_1
\|\qh_{SR}\|^2  }{\Ssr}
\qH_{RR}\qw_t\qw_t^\dag\qH_{RR}^\dag \!+\! \qI \right)^{\!-1}\!\! \qh_{SR}\right)\nonumber\\
 &= \frac{\rho_1}{\Ssr}\left(\|\qh_{SR}\|^2  - \frac{\frac{\kappa \rho_1 \|\qh_{SR}\|^2
}{\Ssr} |\qh_{SR}^\dag\qH_{RR}\qw_t|^2 }{1+ \frac{\kappa \rho_1
\|\qh_{SR}\|^2}{\Ssr} \|\qH_{RR}\qw_t\|^2 }\right).
 \end{align}
Now, the optimization problem in~\eqref{eqn:opt} is re-expressed as
\vspace{-0.2em}
  \begin{align}\label{eqn:opt:with opt w_r}
    \max_{ \|\qw_t\|=1} & \min\left( \frac{\rho_1}{\Ssr}\left(
 \|\qh_{SR}\|^2  \!- \! \frac{ \frac{\kappa \rho_1}{\Ssr}\|\qh_{SR}\|^2 |\qh_{SR}^\dag\qH_{RR}\qw_t|^2 }
{ 1  +   \frac{\kappa \rho_1}{\Ssr}\|\qh_{SR}|^2 \|\qH_{RR}\qw_t\|^2 }\right), \right.\nonumber\\
&\qquad\left.\frac{\kappa \rho_2 }{\Ssr\Srd} \|\qh_{SR}\|^2|
\qh_{RD}\qw_t|^2\right),
\end{align}
which is still difficult to solve due to its nonconvex nature.

 One of the key results of the optimum scheme is presented in the following proposition.
\begin{proposition}
\label{PropositionOptimal}
The optimal $\qw_{t,o}$ is given by
\vspace{-0.5em}
\begin{align}\label{eq:opt-prec1}
\qw_{t,o}=
\begin{cases}
\qw_{\mSINR},&
\\ \qquad \qquad\qquad \text{if}\quad f_1(\qw_{\mSINR})\leq f_2(\qw_{\mSINR})\\
\qw_{\MRT},&
\\ \qquad \qquad\qquad \text{if}\quad f_2(\qw_{\MRT})\leq f_1(\qw_{\MRT})\\
\qw_t~\text{is obtained from the feasibility problem}~{\mathcal P},&
\\ \qquad \qquad\qquad\text{otherwise}
\end{cases}
\end{align}
where
\begin{eqnarray*}
\label{eq:opt-prec1A}
 f_1(\qw_t)&=& \frac{\rho_1}{\Ssr}\left(\|\qh_{SR}\|^2  -  \frac{ \frac{\kappa \rho_1}{\Ssr}\|\qh_{SR}\|^2 |\qh_{SR}^\dag\qH_{RR}\qw_t|^2 }
{ 1  +   \frac{\kappa \rho_1}{\Ssr}\|\qh_{SR}\|^2 \|\qH_{RR}\qw_t\|^2 }\right)\nonumber\\
f_2(\qw_t)&=& \frac{\kappa \rho_2 }{\Ssr\Srd} \|\qh_{SR}\|^2|\qh_{RD}\qw_t|^2\nonumber\\
\end{eqnarray*}
\begin{align}
\qw_{\mSINR}=&\max_{\qw_t}~f_1(\qw_t)\nonumber\\
=&\min_{\qw_t}~\frac{ \frac{\kappa \rho_1}{\Ssr}\|\qh_{SR}\|^2\qw_t^\dag \qH_{RR}^\dag \qh_{SR}\qh_{SR}^\dag \qH_{RR}\qw_t}{\qw_t^\dag\left({\bf I}+ \frac{\kappa \rho_1}{\Ssr}\|\qh_{SR}\|^2\qH_{RR}^\dag \qH_{RR}\right)\qw_t }\nonumber\\
\qw_{\MRT}=&\max_{\qw_t}~f_2(\qw_t)=\frac{\qh_{RD}^{\dag}}{||\qh_{RD}||},\nonumber\\
&~t\in \left[ 0, {\rm min}\left(\frac{ \rho_1}{\Ssr}||\qh_{SR}||^2,   \frac{\kappa\rho_2 }{\Ssr\Srd} \|\qh_{SR}\|^2||\qh_{RD}||^2\right)\right]\nonumber
\end{align}
and ${\mathcal P}$ is
\begin{align}
\label{eq:opt-prec2}
{\mathcal P}:&\hspace{1em}{\rm Find} ~{\qW_t,  y} \nonumber\\
\mbox{s.t.} &\hspace{1em}
ty\leq\!  \frac{ \rho_1}{\Ssr}\left(\|\qh_{SR}\|^2y \right. \nonumber\\
&\hspace{4em}\left.-\frac{ \kappa\rho_1}{\Ssr}\|\qh_{SR}\|^2 \qh_{SR}^\dag\qH_{RR}\qW_t\qH_{RR}^\dag\qh_{SR}\!\right) \nonumber \\
&\hspace{1.5em} y= \mathsf{tr} \left( \qW_t\left(\qI+ \frac{\kappa \rho_1}{\Ssr}\|\qh_{SR}\|^2 \qH_{RR}^{\dag}\qH_{RR}\right)\right)\nonumber\\
&\hspace{1.5em} t\leq  \frac{\kappa \rho_2 }{\Ssr\Srd} \|\qh_{SR}\|^2 \qh_{RD}\qW_t\qh_{RD}^\dag\nonumber\\
 &\hspace{1.5em}\mathsf{tr} \left(\qW_t\right)=1, \qW_t\succeq 0.
\end{align}
The maximum value of $t$ for which the problem ${\mathcal P}$ is feasible gives the relaxed optimum solution ${\qW_t}$.
\end{proposition}
\begin{proof}
Note that $\qw_{\mSINR}$ is the eigenvector corresponding to the minimum eigenvalue of the matrix $(\qI+\frac{\kappa \rho_1}{\Ssr}\|\qh_{SR}\|^2 \qH_{RR}^{\dag} \qH_{RR} )^{-1} (\qH_{RR}^{\dag} \qh_{SR} \qh_{SR}^{\dag} \qH_{RR})$. The solution of $\qw_t$ that maximizes $f_1(\qw_t)$ is $\qw_{\mSINR}$, whereas $f_2(\qw_t)$ is maximized by $\qw_{\MRT}$. Now consider the following cases:
\begin{itemize}
\item $f_1(\qw_{\mSINR})\leq f_2(\qw_{\mSINR})$:  In this case the two possible scenarios for $\qw_t\neq \qw_{\mSINR}$ are $f_1(\qw_t)\leq f_2(\qw_t)$ and  $f_1(\qw_t)>f_2(\qw_t)$. In the former case,  ${\rm min}(f_1(\qw_t), f_2(\qw_t))=f_1(\qw_t)$ but it follows that $f_1(\qw_t)\leq f_2(\qw_{\mSINR})$ since $f_1(\qw_t)\leq f_1(\qw_{\mSINR})$. In the latter case,  ${\rm min}(f_1(\qw_t), f_2(\qw_t))=f_2(\qw_t)$ but it follows that $f_2(\qw_t)\leq f_1(\qw_{\mSINR})\leq f_2(\qw_{\mSINR})$. Therefore, $\qw_{\mSINR}$ is the optimum precoder when
$f_1(\qw_{\mSINR})\leq f_2(\qw_{\mSINR})$.

\item   $f_2(\qw_{\MRT})\leq f_1(\qw_{\MRT})$:  The two possible scenarios are $f_2(\qw_t)\leq f_1(\qw_t)$ and  $f_2(\qw_t)>f_1(\qw_t)$ where  $\qw_t\neq \qw_{\MRT}$. In the former case,  ${\rm min}(f_1(\qw_t), f_2(\qw_t))=f_2(\qw_t)$ but it follows that $f_2(\qw_t)\leq f_1(\qw_{\MRT})$ since  $f_2(\qw_t)\leq f_2(\qw_{\MRT})$. In the latter case,  ${\rm min}(f_1(\qw_t), f_2(\qw_t))=f_1(\qw_t)$ but it follows that $f_1(\qw_t)\leq f_2(\qw_{\MRT})\leq f_1(\qw_{\MRT})$. Therefore, $\qw_{\MRT}$ is the optimum precoder when
$f_2(\qw_{\MRT})\leq f_1(\qw_{\MRT})$.

\item For remaining cases, to the best of our knowledge, there is no closed-form solution for the optimization problem  \eqref{eqn:opt:with opt w_r}. Introducing an auxiliary variable $t\geq 0$,  \eqref{eqn:opt:with opt w_r} is expressed as
\begin{align}\label{eqn: nonconvex quadratic optimization}
&\max_{\|\qw_t\|=1, t}\hspace{2em} t \nonumber\\
&\hspace{1em}\mbox{s.t.}\hspace{1em}  t \leq  \frac{\rho_1}{\Ssr}\left(\|\qh_{SR}\|^2 \!\! -\!\!  \frac{ \frac{\kappa \rho_1}{\Ssr}\|\qh_{SR}\|^2 |\qh_{SR}^\dag\qH_{RR}\qw_t|^2 }
                 { 1+ \frac{\kappa \rho_1}{\Ssr}\|\qh_{SR}\|^2 \|\qH_{RR}\qw_t\|^2 }\right)\nonumber \\
&\hspace{3.4em} t\leq  \frac{\kappa \rho_2 }{\Ssr\Srd} \|\qh_{SR}\|^2|\qh_{RD}\qw_t|^2.
\end{align}
This is a nonconvex quadratic optimization problem with nonconvex constraint. To solve the problem in~\eqref{eqn: nonconvex quadratic optimization}, we first apply SDR technique by using  a positive-semidefinite matrix  ${\qW}_{t} = \qw_t\qw_t^{\dag}$ and relaxing the rank-constraint on ${\qW}_{t}$. Moreover, define the auxiliary variable
\begin{equation}
y= \mathsf{tr} \left(\qW_t \left(\qI+ \frac{\kappa \rho_1}{\Ssr}\|\qh_{SR}\|^2 \qH_{RR}^{\dag}\qH_{RR}\right)\right).\nonumber
\end{equation}
The relaxed optimization problem~\eqref{eqn: nonconvex quadratic optimization} in terms of $\qW_t$, $y$, and $t$ is
\vspace{-0.2em}
\bea\label{eqn: nonconvex quadratic optimization 2}
\max_{ \qW_t,  t, y} &&\!\!\!\!\!\!\!\! t \nonumber\\
\mbox{s.t.} &&\!\!\!\!\!\!\!\! ty\leq  \frac{\rho_1}{\Ssr}\left(\|\qh_{SR}\|^2  y \right.\nonumber\\
&&\left.\qquad- \frac{\kappa \rho_1}{\Ssr}\|\qh_{SR}\|^2 \qh_{SR}^\dag\qH_{RR}\qW_t\qH_{RR}^\dag\qh_{SR}\right) \nonumber \\
&&\!\!\!\!\! y= \mathsf{tr} \left(\qW_t \left(\qI+ \frac{\kappa \rho_1}{\Ssr}\|\qh_{SR}\|^2 \qH_{RR}^{\dag}\qH_{RR}\right)\right)\nonumber\\
&&\!\!\!\!\! t\leq \frac{\kappa \rho_2}{\Ssr\Srd} \|\qh_{SR}\|^2\qh_{RD} \qW_t \qh_{RD}^{\dag}\nonumber\\
&&\!\!\!\!\!\mathsf{tr} \left(\qW_t\right)=1, \qW_t\succeq 0.
\eea
The optimization problem~\eqref{eqn: nonconvex quadratic optimization 2} is still nonconvex. However, for a given $t$, the problem turns to a convex feasibility problem given in \eqref{eq:opt-prec2}. We show that there is no need to solve this feasibility problem for all possible values of $t$. This can be demonstrated as follows. Note that the constraints of the problem~\eqref{eqn: nonconvex quadratic optimization 2} show that $t$ can be upper bounded as $t \leq {\rm min}\left(\frac{\rho_1}{\Ssr}||\qh_{SR}^2||^2, \frac{\kappa \rho_2 }{\Ssr\Srd} \|\qh_{SR}\|^2||\qh_{RD}||^2\right)$. Consequently, we can start solving the feasibility problem  \eqref{eq:opt-prec2} in the decreasing order for $t$. The largest value of $t$ for which the problem is feasible yields the optimum $\qW_t$. If the optimum $\qW_t$ is rank-one, then the relaxed problem~\eqref{eqn: nonconvex quadratic optimization 2} is equivalent to the original problem.
\end{itemize}
\end{proof}
In the following, we show that the relaxed optimization \eqref{eqn: nonconvex quadratic optimization 2}  or convex feasibility  problem~\eqref{eq:opt-prec2} provides an optimum rank-one solution. To this end, we present another key result of the optimum scheme in the following proposition.
\hspace*{-0.3cm}\begin{proposition}
\label{ProofPropositionRankOne}
The rank-one optimum $\qW_t$ can  always be guaranteed in \eqref{eqn: nonconvex quadratic optimization 2}.
\end{proposition}
\begin{proof}
The proof is based on Karush-Kuhn-Tucker (KKT) conditions and given in Appendix I.
\end{proof}

\subsection{TZF Scheme}
We now present some suboptimum beamforming solutions. The first is the TZF scheme, where relay takes advantage of the multiple transmit antennas to completely cancel the LI~\cite{Suraweera:TWC:2014}. To ensure this is feasible, the number of the transmit antennas at relay should be greater than one, i.e., $M_T>1$. In addition, MRC is applied at the relay input, i.e., $\qw_r =
\frac{\qh_{SR}}{\|\qh_{SR}\|}$. After substituting $\qw_r$ into \eqref{eqn:opt}, the optimal transmit beamforming vector $\qw_t$ is obtained by solving the following problem:
\vspace{-0.2em}
   \bea\label{eqn:wt}
    \max_{\|\qw_t\|=1} &&\hspace{1em}  | \qh_{RD} \qw_t|^2 \nonumber\\
     \mbox{s.t.} &&\hspace{1em} \qh_{SR}^\dag\qH_{RR}\qw_t =0.
 \eea
 We know that ${\qA}\triangleq \qH_{RR}^\dag  \qh_{SR} \qh_{SR}^\dag \qH_{RR}$ is a rank-one Hermitian matrix with eigenvalue $\lambda \triangleq \|\qh_{SR}^\dag\qH_{RR}\|^2$ and eigenvector $\qx\triangleq \frac{\qH_{RR}^\dag \qh_{SR}}{|| \qH_{RR}^\dag \qh_{SR}||}$. Consequently, the eigenvalue decomposition of $\qA$  can be given by $\qx^\dag \left(\qI-\frac{1}{\lambda} \qA\right)={\bf 0}$ which implies that $\qx^\dag \left(\qI-\frac{1}{\lambda} \qA\right){\bar \qw}_t=0$ for all ${\bar\qw}_t \neq {\bf 0}$. Comparing this with the ZF constraint in \eqref{eqn:wt}, it is clear that we can take $\qw_t=\qB {\bar \qw}_t$, where $\qB\triangleq \qI -\frac{\qH_{RR}^\dag\qh_{SR}\qh_{SR}^\dag\qH_{RR}}{\|\qh_{SR}^\dag\qH_{RR}\|^2}$, without violating the ZF constraint. As such,  the objective function in  \eqref{eqn:wt} reduces to $|\qh_{RD} \qB{\bar \qw}_t|^2$ which is maximized with ${\bar \qw}_t=k_c \qB\qh_{RD}^\dag$. Since $||\qw_t||=1$ and $\qB=\qB^2$, it is clear that $k_c=\frac{1}{ ||\qB \qh_{RD}^\dag||}$. Consequently, the transmit beamformer for the TZF scheme is given by
 \vspace{-0.2em}
\bea\label{eqn:wtZF}
\qw_t^{\ZF} = \frac{\qB \qh_{RD}^{\dag}}{\|\qB\qh_{RD}^{\dag}\|}.
 \eea
\vspace{-0.8em}
\subsection{RZF Scheme}
As an alternative solution, the transmit beamforming vector can be set using the MRT principle, i.e., $\qw_t =\frac{\qh_{RD}^{\dag}}{\|\qh_{RD} \|}$, and $\qw_r$ is designed with the ZF criterion $\qw_r^\dag \qH_{RR} \qw_t=0$. To ensure feasibility of RZF, $R$ should be equipped with $M_R>1$ receive antennas. Substituting the MRT solution for $\qw_t$ into  \eqref{eqn:opt}, the optimal receive beamforming vector $\qw_r$ is the solution of the following problem:
   \bea\label{eqn:wr}
    \max_{\|\qw_r\|=1} &&\hspace{1em} |\qw_r^\dag \qh_{SR}|^2\nonumber\\
     \mbox{s.t.}&&\hspace{1em}  \qw_r^\dag \qH_{RR}\qh_{RD}^\dag =0.
 \eea
Using similar steps as in the TZF scheme, the optimal combining vector $\qw_r$  is obtained as
\vspace{-0.2em}
\bea\label{eqn:wrZF}
\qw_r^{\ZF} = \frac{\qD \qh_{SR}}{\|\qD\qh_{SR}\|},
 \eea
where $\qD\triangleq  \qI - \frac{\qH_{RR}\qh_{RD}^\dag\qh_{RD} \qH_{RR}^\dag}{\|\qH_{RR}\qh_{RD}\|^2}$.
\subsection{MRC/MRT Scheme}
Finally, we consider the MRC/MRT scheme, where $\qw_r$ and $\qw_t$ are set to match the first hop and second hop channel, respectively. Hence,
\vspace{-0.2em}
\bea\label{eqn:wrMRC/MRT}
\qw_r^{\MRC}  = \frac{\qh_{SR}^{\dag}}{\| \qh_{SR} \|},\quad\qw_t^{\MRT} =\frac{\qh_{RD}^{\dag}}{\|\qh_{RD} \|}.
 \eea
It is worthwhile to note that the optimum, TZF, and RZF schemes reduce to the MRC/MRT scheme in the absence of LI. Although the MRC/MRT scheme is not optimal in the presence of LI, it could be favored in situations where compatibility with HD systems is a concern. Moreover, as we show in Section VI, the MRC/MRT scheme exhibits a very good performance as compared to other schemes under mild LI effect. Note that the MRC/MRT scheme requires only the  knowledge of $\qh_{SR}$ and  $\qh_{RD}$, whereas the other three schemes require the knowledge of  $\qh_{SR}$, $\qh_{RD}$, and $\qH_{RR}$.

\section{Optimizing $\alpha$ for Instantaneous Throughput}\label{sec:IT}
In this section, we optimize $\alpha$ for the FD beamforming schemes proposed in the previous section and maximize the instantaneous throughput which is given by~\cite{Caijun:TCOM:2014}
\begin{align}
R_{\mathsf{I}}(\alpha) =(1-\alpha)\log_2 (1+\gamma_{\mathsf{FD}}).
\end{align}
This expression reveals an interesting trade-off between the duration of energy harvesting and the instantaneous throughput. A longer energy harvesting time increases the harvested energy and consequently the second hop SNR, however decreases the available time for information transmission and vice-versa. Therefore, an appropriate system design can optimize the instantaneous throughput by adjusting $\alpha$.

 \subsection{Optimum Scheme}
Substituting the optimum $\qw_{t,o}$, the instantaneous rate as a function of $\alpha$ is
  \begin{align}\label{eqn:opt:with opt alpha 1}
R_{\rm Opt}(\alpha)
     &=(1-\alpha) \log_2\left( 1+ \|\qh_{SR}\|^2\right.\\
      &\left.\times\min\left(\! \frac{\rho_1}{\Ssr}\left(
1  \!-\!  \frac{ \frac{\kappa \rho_1}{\Ssr} |\qh_{SR}^\dag\qH_{RR}\qw_{t,o}|^2 }
{ 1  \!+ \!  \frac{\kappa \rho_1}{\Ssr}\|\qh_{SR}\|^2 \|\qH_{RR}\qw_{t,o}\|^2 }\right), \right.\right.\nonumber\\
&\left.\left.\hspace{7em}\frac{\kappa \rho_2 }{\Ssr\Srd}|
\qh_{RD} \qw_{t,o}|^2\!\right)\!\right)\!.\nonumber
\end{align}
Let $b_0= \frac{\rho_2}{\rho_1}\frac{\eta}{\Srd}|
\qh_{RD}\qw_{t,o}|^2$, $b_1=\frac{\eta \rho_1}{\Ssr} |\qh_{SR}^\dag\qH_{RR}\qw_{t,o}|^2$, and $b_2=\frac{\eta \rho_1}{\Ssr}\|\qh_{SR}\|^2 \|\qH_{RR}\qw_{t,o}\|^2$.  Hence,~\eqref{eqn:opt:with opt alpha 1} is written as
 \begin{align}
 \label{eqn:opt:with opt alpha 2}
    R_{\rm Opt}(\alpha)
      &= (1-\alpha) \log_2\left( 1+ \frac{\rho_1}{\Ssr}\|\qh_{SR}\|^2\right.\nonumber\\
      &\left.\times\min\left(1-\frac{ \frac{\alpha b_1}{1-\alpha}}{1+\frac{\alpha b_2}{1-\alpha}}, \frac{\alpha b_0}{1-\alpha}\right)\right).
 \end{align}
Thus, the optimal $\alpha$ is obtained by solving
\begin{align}
\alpha^{*}_{\rm Opt}  = \argmax_{0<\alpha<1}  R_{\rm opt}(\alpha).
\end{align}
Denote $f= \frac{\rho_1}{\Ssr}\|\qh_{SR}\|^2$ and ${\tilde f}=fb_0$. The above optimization problem can be solved analytically as shown in the following proposition.
\begin{proposition}
The optimal $\alpha^{*}_{\rm Opt}$ is given by
\begin{align}\label{eq:opt alpha opt 3}
\alpha^{*}_{\rm Opt}=
\begin{cases}
\frac{e^{W\left(\frac{{\tilde f} -1}{e}\right)+1}-1}{{\tilde f} -1 + e^{W\left(\frac{{\tilde f}-1}{e}\right)+1}},&
\!\!\!\text{if}\quad e^{W\left(\frac{{\tilde f} -1}{e}\right)+1}<\alpha_0{\tilde f}+1; \\
\frac{\alpha_0}{1+\alpha_0},&\!\!\!\text{otherwise},
\end{cases}
\end{align}
where $W(x)$ is the Lambert $W$ function in which $W(x)$ is the solution of $W \exp(W) = x$ and noting that $b_1 \leq b_2$ from Cauchy-Schwarz inequality, $\alpha_0$ is given by
\begin{eqnarray*}
\label{eqn:opt:with opt alpha 4}
\alpha_0
= \frac{(b_2-b_1-b_0)+\sqrt{ b_0^2+(b_2-b_1)^2+2b_0(b_1+b_2)}}{2b_0b_2}.
\end{eqnarray*}
\end{proposition}

\begin{proof}
The following two cases are considered:
\begin{enumerate}
 \item if $ \frac{\alpha b_0}{1-\alpha}<1-\frac{ \frac{\alpha b_1}{1-\alpha}}{1+\frac{\alpha b_2}{1-\alpha}}$ or $\alpha < \frac{\alpha_0}{1+\alpha_0}$, we have
     \vspace{-0.2em}
 \bea
 \label{eqn:opt:with opt alpha 5}
    R_{\rm Opt}(\alpha)
    = (1-\alpha)\log_2 \Big(1+  \frac{\alpha}{1-\alpha} {\tilde f}\Big).\nonumber
 \eea
Therefore, taking the first order derivative of $R_{\rm Opt}(\alpha)$ with respect to $\alpha$, and using the procedure described in~\cite{Caijun:TCOM:2014}, the optimal time portion $\alpha$ can be obtained as
\bea
\label{eqn:opt:with opt alpha 5}
    \alpha_{\rm Opt}^* = \frac{e^{W(\frac{{\tilde f} -1}{e})+1}-1}{{\tilde f} -1 + e^{W(\frac{{\tilde f}-1}{e})+1}}.
  \eea
    \item  Otherwise, if $\alpha > \frac{\alpha_0}{1+\alpha_0}$, the instantaneous throughput in~\eqref{eqn:opt:with opt alpha 2} is given by
    \vspace{-0.3em}
 \bea
    R_{\rm Opt}(\alpha)
= (1-\alpha)\log_2 \Big(1\!+\! f\Big( 1- \frac{\alpha b_1}{1-\alpha+\alpha b_2}\Big)\Big).\nonumber
 \eea

Taking the first derivative of $R_{\rm Opt}(\alpha)$ with respect to $\alpha$ yields
\bea
    \frac{ d R_{\rm Opt}(\alpha)}{d \alpha}
    =-\log_2 \Big(1+ f\Big( 1- \frac{\alpha b_1}{1-\alpha+\alpha b_2}\Big)\Big)\nonumber\\
    -\frac{(1-\alpha)f b_1}{(1-\alpha+\alpha b_2)^2 \log2}\frac{1}{ 1+ f\Big( 1- \frac{\alpha b_1}{1-\alpha+\alpha b_2}\Big)},\nonumber
 \eea
 which is strictly smaller than zero. Therefore, $R_{\rm Opt}(\alpha)$ is a decreasing function with respect to $\alpha$, and hence the optimal $\alpha$ is given by
 \bea
    \alpha_{\rm Opt}^* &=& \frac{\alpha_0}{1+\alpha_0}.
  \eea

  \end{enumerate}
 \end{proof}
\subsection{TZF Scheme}
Substituting $\qw_t^{\ZF}$ into~\eqref{eq: e2e snr general}, the instantaneous throughput of the TZF scheme is given by
 \bea
    R_{\mathsf{I},\TZF}(\alpha)
    = (1-\alpha)\log_2 \left(1+ \min\left(  \frac{\rho_1}{\Ssr} \|\qh_{SR}\|^2,\right.\right.\nonumber\\
     \left.\left.\frac{\kappa \rho_2}{\Ssr\Srd} \|\qh_{SR}\|^2\|\qB\qh_{RD}\|^2
    \right)\right).
 \eea
 Hence, the optimal $\alpha $ can be obtained by solving the following optimization problem
\begin{align}\label{eq:opt:TZF}
&\alpha^{*}_{\TZF}  = \argmax_{0<\alpha<1}  R_{\mathsf{I},\TZF}(\alpha).
\end{align}
In~\eqref{eq:opt:TZF}, since $R_{\mathsf{I},\TZF}(\alpha)$ is concave with respect to $\alpha$, $\alpha^{*}_{\TZF}$ can be obtained by solving the equation $\frac{d R_{\mathsf{I},\TZF}(\alpha)}{d \alpha}=0$. Using the similar approach as in the optimum scheme, $\alpha^{*}_{\TZF}$ can be derived as
\begin{align}\label{eq:opt-alpha-TZF}
\alpha^{*}_{\TZF}=
\begin{cases}
\frac{e^{W\left(\frac{a_1 -1}{e}\right)+1}-1}{a_1 -1 + e^{W\left(\frac{a_1-1}{e}\right)+1}},& \text{if}\quad e^{W\left(\frac{a_1 -1}{e}\right)+1}<\frac{a_1}{\alpha_{1}} +1;\\
\frac{1}{1 + \alpha_{1}},&\text{otherwise},
\end{cases}
\end{align}
where $a_1 =  \frac{ \eta \rho_2}{\Ssr\Srd}\|\qh_{SR}\|^2\|\qB\qh_{RD}\|^2$, $a_2  = \frac{\Ssr}{ \rho_1 \|  \qh_{SR}\|^2}$, and $\alpha_{1} = a_1 a_2$.
\vspace{-0.2em}
\subsection{RZF Scheme}
Substituting $\qw_r^{\ZF}$ into~\eqref{eq: e2e snr general}, the instantaneous throughput of the RZF scheme can be written as
\vspace{-0.2em}
 \bea
    R_{\mathsf{I},\RZF}(\alpha)
    = (1-\alpha)\log_2 \Big(1+ \min\Big(  \frac{\rho_1}{\Ssr} \|\qD\qh_{SR}\|^2,\nonumber\\
      \frac{\kappa \rho_2}{\Ssr\Srd} \|\qh_{SR}\|^2\|\qh_{RD}\|^2
    \Big)\Big).
 \eea
Accordingly, the optimal $\alpha $ can be obtained as
\begin{align}\label{eq:opt-alpha-RZF}
\alpha^{*}_{\RZF}=
\begin{cases}
\frac{e^{W\left(\frac{a_3 -1}{e}\right)+1}-1}{a_3 -1 + e^{W\left(\frac{a_1-1}{e}\right)+1}},& \text{if}\quad e^{W\left(\frac{a_3 -1}{e}\right)+1}<\frac{a_3}{\alpha_{2}} + 1;\\
\frac{1}{1 + \alpha_{2}},&\text{otherwise},
\end{cases}
\end{align}
where $a_3 =  \frac{ \eta \rho_2}{\Ssr\Srd}\|\qh_{SR}\|^2\|\qh_{RD}\|^2$, $a_4  = \frac{\Ssr}{ \rho_1 \|\qD  \qh_{SR}\|^2}$, and $\alpha_{2} = a_3 a_4$.
\vspace{-0.5em}
\subsection{MRC/MRT Scheme}
Substituting $\qw_r^{\MRC}$ and $\qw_t^{\MRT}$ into~\eqref{eq: e2e snr general}, the instantaneous throughput of the MRC/MRT scheme can be expressed as
\vspace{-0.2em}
\begin{align} \label{eq:R_MRC_initial}
    &R_{\mathsf{I},\MRC}(\alpha)
    = (1\!-\!\alpha)\log_2 \left(1+ \right.\nonumber\\
    &\hspace{5em}\left.\min\left( \frac{\frac{\rho_1}{\Ssr}\| \qh_{SR}\|^2}{ \frac{\kappa \rho_1}{\Ssr}
\|\qh_{SR}\|^2 \frac{|\hSR^{\dag}\vH_{R R}\hRD^{\dag}|^2}{\|\hSR \hRD\|^2}  +1},\right.\right.\nonumber\\
     &\hspace{7em}\left.\left.\frac{\kappa\rho_2}{\Ssr\Srd} \|\qh_{SR}\|^2\|
\qh_{RD}\|^2
    \right)\right).
\end{align}
For the notational convenience, we denote $b_3= \frac{\rho_2}{\Ssr\Srd}\| \qh_{SR}\|^2\| \qh_{RD}\|^2$, $b_4  = \frac{\rho_2}{\Ssr\Srd} |\hSR^{\dag}\vH_{R R}\hRD^{\dag}|^2$, and $b_5 = \frac{\rho_2}{\rho_1}\frac{1}{\Srd} \|\hRD\|^2$. Hence,~\eqref{eq:R_MRC_initial} can be written as
\vspace{-0.2em}
\begin{align}\label{eq:R_MRC_notational}
    R_{\mathsf{I},\MRC}(\alpha)
      &=(1-\alpha)\log_2 \left(1 + b_3\right.\nonumber\\
      &\hspace{2em}\left.\times  \min\left(
    \frac{1}{ \frac{\eta\alpha}{1-\alpha} b_4 +
b_5},\frac{\eta\alpha}{1\!-\!\alpha}
    \!\right)\!\right)\!.
\end{align}
As such, the optimal time-split $\alpha$ is the solution of the following problem:
\begin{align}
&\alpha^{*}_{\MRC}  = \argmax_{0<\alpha<1}  R_{\mathsf{I},\MRC}(\alpha).
\end{align}
The above optimization problem can be solved analytically, and we have the following key result.
\begin{proposition}
The optimal $\alpha^{*}_{\MRC}$ is given by
\vspace{-0.3em}
\begin{align}\label{eq:opt-alpha-MRC}
\alpha^{*}_{\MRC}=
\begin{cases}
\frac{e^{W\left(\frac{a_3 -1}{e}\right)+1}-1}{a_3 -1 + e^{W\left(\frac{a_3-1}{e}\right)+1}},& \text{if}\quad e^{W\left(\frac{a_3 -1}{e}\right)+1}<\frac{a_3}{\alpha_{3}} +1;\\
\frac{1}{1 + \alpha_{3}},&\text{otherwise},
\end{cases}
\end{align}
 where $\alpha_3 = \frac{2 \eta b_4}{-b_5 + \sqrt{b_5^2 + 4b_4}}$.
\end{proposition}
\begin{proof}
We consider two cases as follows:
\begin{enumerate}
 \item if $\frac{\eta\alpha}{1-\alpha} < \frac{1}{ \frac{\eta\alpha}{1-\alpha} b_4 +
b_5}$  or $\alpha < \frac{1}{1 + \alpha_3}$, we have
 \bea
    R_{\mathsf{I},\MRC}(\alpha)
    &= &(1-\alpha)\log_2 \left(1+  \frac{\eta\alpha}{1-\alpha}b_3\right).\nonumber
 \eea

Therefore, taking the first order derivative of $R_{\mathsf{I},\MRC}(\alpha)$ with respect to $\alpha$, and following the same procedure as in optimum scheme, the optimal time portion $\alpha$ can be obtained as
\vspace{-0.2em}
\bea
    \alpha^* &=& \frac{e^{W(\frac{a_3 -1}{e})+1}-1}{a_3 -1 + e^{W(\frac{a_3-1}{e})+1}},
  \eea
where $a_3 = \frac{\eta \rho_2}{\Ssr\Srd}\| \qh_{SR}\|^2\| \qh_{RD}\|^2 = \eta b_3$.
 \item  Otherwise, if $\alpha > \frac{1}{1 + \alpha_3}$, the instantaneous throughput in~\eqref{eq:R_MRC_notational} is given by
 \vspace{-0.2em}
 \bea
    R_{\mathsf{I},\MRC}(\alpha)
    = (1-\alpha)\log_2 \left(1+  \frac{b_3}{ \frac{\eta\alpha}{1 - \alpha}b_4 + b_5}\right).\nonumber
 \eea

 Taking the first derivative of $R_{\mathsf{I},\MRC}(\alpha)$ with respect to $\alpha$ yields
 \vspace{-0.3em}
\bea
    \frac{ d R_{\mathsf{I},\MRC}(\alpha)}{d \alpha}
    = -\log_2 \left(1+  \frac{b_3}{ \frac{ \eta\alpha}{1 - \alpha}b_4 + b_5}\right) \nonumber\\
    - \frac{(1-\alpha)\eta b_3b_4}{\left(\frac{ \eta\alpha}{1 - \alpha}b_4 + b_5\right) \left(\frac{ \eta\alpha}{1 - \alpha}b_4 + b_5 + b_3\right)\log2}.\nonumber
 \eea
 which is strictly smaller than zero. Therefore, $R_{\mathsf{I},\MRC}(\alpha)$ is a decreasing function with respect to $\alpha$, and hence the optimal $\alpha$ is given by
 \vspace{-0.3em}
 \bea
    \alpha^* &=& \frac{1}{1 + \alpha_3}
  \eea
  \end{enumerate}
 \end{proof}
We end this section with the following remarks. In contrast to the suboptimum schemes,  the obtained solution of $\qw_t$  in optimum scheme depends on $\alpha$. As a consequence, joint optimization w.r.t. $\alpha$ and $\qw_t$ is required in the latter scheme. There are two ways to solve this joint optimization. The first way is to find $\qw_t$  by following Proposition 1 and next performing a one-dimensional line search over $0<\alpha<1$. This guarantees the global optimum solutions for $\alpha$ and $\qw_t$. Another way is to employ an iterative approach  where each iteration step consists of a two-step  optimization, i.e., optimizing $\qw_t$ for a given $\alpha$ and vice-versa. In either way, the computational complexity of the proposed optimum scheme is limited due to the following facts. As seen from Proposition 1, there is a need to solve SDR problem only for a specific scenario. Moreover,  the  rank-one optimum solution can always be recovered from optimum $\qW_t$ which is obtained by solving the convex feasibility problem ${\mathcal P}$ in a few number of iterations. Also $\alpha$ can be obtained analytically for a given $\qw_t$ or only a  one-dimensional search is  required for finding the optimum $\alpha$. Despite these facts,  the computational complexity  of the optimum scheme is higher than that of the suboptimum schemes. This may be justified since, depending on the scenarios,  the optimum scheme  significantly outperforms the  suboptimum schemes (see Fig. \ref{fig: Instantanous-throu-v-alpha}).
 \begin{algorithm}
 \caption{ The proposed optimum scheme for the instantaneous throughput maximization.}
 \begin{algorithmic}
  \renewcommand{\algorithmicrequire}{\textbf{Step 1:} }
 \REQUIRE Initialize $\alpha$:\\
 Choose $\alpha$ from its grid (Line search Method -LS) or\\
  generate an initial point for $\alpha$ (Alternating  optimization method-AO) where $\alpha \in [0, 1)$.
 \renewcommand{\algorithmicrequire}{\textbf{Step 2:} }
 \REQUIRE Obtain the transmit beamformer ${\bf w}_{t,{\rm o}}$ using (15).
  \IF {(15) requires solving the feasibility problem ${\mathcal P}$ of (16)}
  \STATE \vspace{-1em}
    \IF { ${\bf W}_{\rm t}$ is rank-one}
  \STATE Take ${\bf w}_{t, {\rm o}}$ as the eigenvector corresponding to non-zero eigenvalue of  ${\bf W}_{\rm t}$.
 \ELSE
    \STATE Take ${\bf w}_{t, {\rm o}}$ as given in {\it Case b} of the Appendix I  ( see (82))
 \ENDIF
  \ENDIF
 \renewcommand{\algorithmicrequire}{\textbf{Step 4:} }
 \REQUIRE
 \IF {AO method}
 \WHILE{not converged }
  \STATE Obtain $\alpha_{\rm Opt}$ from (28) using ${\bf w}_{t, {\rm o}}$ obtained in\textbf{ Step 2}.\\
  \STATE Update $\alpha$ as $\alpha=\alpha_{\rm Opt}$ and go to \textbf{Step 2}.
 \ENDWHILE
 \STATE Save ${\bf w}_{t, {\rm o}}$, $\alpha$,  and the objective function.\\
 \ELSIF{LS method}
    \STATE Save ${\bf w}_{t, {\rm o}}$, $\alpha$,  and the objective function.\\
    \STATE Take another $\alpha$  from its grid, and repeat \textbf{Step 1}.
 \ENDIF

 \renewcommand{\algorithmicrequire}{\textbf{Step 3:} }
 \REQUIRE Choose the  ${\bf w}_{t, {\rm o}}$ and $\alpha$ that give maximum objective value.
 \end{algorithmic}
 \end{algorithm}

We outline the proposed optimum scheme for the instantaneous throughput maximization problem in Algorithm 1.
\bigformulatop{46}{ \vskip-0.5cm
\begin{align}\label{eq:cdf of gammaTZF: Asymptotic}
F_{\gamma_{\TZF}^{\infty}}(z)& \approx
 \begin{cases}
 \left(\frac{1}{\Gamma(M_R+1)}
 +
  \frac{1}{\Gamma(M_T-1)\Gamma(M_R)}
  \!\sum\limits_{k=0}
       ^{\infty}
 \frac{(-1)^{k+1}}{k!(k+M_T)}
  \left(\frac{\SnD}{\SnR}\frac{d_2^{\tau}}{\kappa} \right)^{M_T+k-1}
  \!\!\!\!\!\frac{1}{M_R - M_T-k +1}
   \right)
 \left(\bGsd\right)^{M_R}\!\!\!\!, &                                M_T>M_R+\!1,  \\
\frac{1}{\Gamma(M_R+1)}\left(1+\frac{1}{\Gamma(M_R)}
        \left(\ln\left({\rho_1}\right)-
        \ln\left({ d_1^{\tau}z }\right) + \psi(1)\right) \left(\frac{\SnD}{\SnR}\frac{ \Srd}{\kappa} \right)^{M_R}  \right)
        \left(\bGsd\right)^{M_R},
         & M_T = M_R+\!1,\\
        \frac{\Gamma(M_R - M_T + 1)}{\Gamma(M_T)\Gamma(M_R)}\left(\frac{d_2^{\tau}}{\kappa} \right)^{M_T-1}
        \left(\frac{\Ssr z}{\rho_2}\right)^{M_T -1}, &  M_T < M_R+\!1.
\end{cases}
\end{align}
}
\section{Delay-constrained Throughput}\label{sec:DCT}
We now consider the delay-constrained scenario, where the source transmits at a constant rate $R_c$ bits/sec/Hz. Due to the time variation of the fading channel, outage events where the instantaneous channel capacity is below the source transmission rate may occur. Hence, the average throughput can
be computed as~\cite{Nasir:TWC:2013}
\begin{align}\label{eq:Delay-constrained Transmission}
R_{\mathsf{D}}(\alpha) = (1-\Pout)R_c (1-\alpha),
\end{align}
where $\Pout$ is the outage probability, which is defined as the probability that the instantaneous SINR falls below a predefined threshold, $\gamth$. Mathematically, it can be written as
\vspace{-0.2em}
\begin{align}\label{eq:OP definition}
\Pout = \Prob(\gamma_{\mathsf{FD}} < \gamth) = F_{\gamma}(\gamth),
\end{align}
where $\gamth = 2^{R_c}-1$. Therefore in order to find the delay-constrained throughput, the remaining key task is to characterize the exact outage probability of the system. In the sequel, we investigate the outage probability of the considered TZF, RZF, and MRC/MRT schemes. In addition, simple high SNR approximations are presented, which provide a concise characterization of the diversity order, and enable a performance comparison of these three schemes. Derivation of the outage probability of the optimum scheme is difficult. Hence we have resorted to simulations for evaluating the delay-constrained throughput of the optimum scheme in Section~\ref{sec:numerical results}.
\vspace{-0.4em}
\subsection{TZF Scheme}
Substituting the $\qw_t^{\ZF}$ and $\qw_r^{\MRC}$ into~\eqref{eq: e2e snr general}, the end-to-end SNR $\gamma_{\TZF}$ can be expressed as
\vspace{-0.2em}
\begin{align}\label{eq:gammaTZF}
\gamma_{\TZF}\!=\!\frac{\rho_1 \|\hSR\|^2}{(1-\alpha) d_1^{\tau}} \min\left(1\!-\alpha,\! \frac{\rho_2}{\rho_1}\frac{\eta\alpha}{d_2^{\tau}}\|\thRD\|^2\right)\!,
\end{align}
where $\thRD$ is an $(M_T-1) \times 1$ vector. $\|\thRD\|^2$ follows the chi-square distribution with $(M_T-1)$ degrees of freedom, denoted as $\|\thRD\|^2 \sim \chi_{2(M_T-1)}^2$~\cite{MathematicalStatistics_1978}. Let $Y_1 = \min\left(1-\alpha,  \frac{\rho_2}{\rho_1}\frac{\eta\alpha}{d_2^{\tau}}\|\thRD\|^2 \right)$. The cdf of $Y_1$ is given by~\cite[Appendix II]{Zhu:TCOM:2015}
\vspace{-0.2em}
\begin{align}\label{eq:cdf of gammaTZF: Y}
F_{Y_1} (y)&= \left\{%
 \begin{array}{clcr}
  1 &                  y > 1-\alpha,  \\
  1 - \frac{\Gamma\left(M_T-1,\frac{\rho_1}{\rho_2}\frac{ y\Srd}{\eta \alpha}\right)}{\Gamma(M_T-1)}              &      y < 1-\alpha.
  \end{array}%
\right.
\end{align}
 Therefore, the cdf of $\gamma_{\TZF}$ can be obtained as
\vspace{-0.2em}
\begin{align}
\label{eq:cdf of gammaTZF}
F_{\gamma_{\TZF}} (z)
  &=\!
  1 \!-\! \frac{1}{\Gamma(M_R)}\int_{\frac{ d_1^{\tau}z }{\rho_1}}^{\infty}\!\!
  Q\left(M_T\!-1,\frac{d_1^\tau d_2^\tau}{\kappa \rho_2} \frac{z}{x}\right)\nonumber\\
  &\qquad\qquad\qquad\times
   x^{M_R\!-1}e^{-x} dx,
\end{align}
where $Q(a,x) = \Gamma(a,x)/\Gamma(a)$.

To the best of the authors's knowledge, the integral in~\eqref{eq:cdf of gammaTZF} does not admit a closed-form expression. However,~\eqref{eq:cdf of gammaTZF} can be evaluated numerically.

To gain further insights, we now look into the high SNR regime and derive a simple approximation for the outage probability, which enables the characterization of the achievable diversity order of the TZF scheme.
\begin{proposition}\label{Propos:hsnr:TZF}
In the high SNR regime, i.e., $\rho_1,\rho_2 \rightarrow \infty$, the outage probability of the TZF scheme can be approximated as~\eqref{eq:cdf of gammaTZF: Asymptotic} at the top of the page.

\end{proposition}
\begin{proof}
See Appendix~\ref{propos:Apx:TZF:highSNR}.
\end{proof}

By inspecting~\eqref{eq:cdf of gammaTZF: Asymptotic}, we see that the TZF scheme achieves a diversity order of $\min(M_R, M_T-1)$. This is intuitive since one degree of freedom is used for interference cancellation. Moreover, we notice that for the case $M_R+1=M_T$, $F_{\gamma_{\TZF}^{\infty}}(z)$ decays as $\rho_1^{-M_R}\ln(\rho_1)$ rather than $\rho_1^{-M_R}$ as in the conventional case, which implies that in the energy harvesting case the slope of $F_{\gamma_{\TZF}^{\infty}}(z)$ converges much slower compared with that in the constant power case.
\vspace{-0.3em}
\subsection{RZF Scheme}
Invoking~\eqref{eq: e2e snr general}, and using $\qw_r^{\ZF}$ and $\qw_t^{\MRT}$, the end-to-end SNR $\gamma_{\RZF}$ can be expressed as
\vspace{-0.3em}
\begin{align}\label{eqn:SNR:R2}
    &\gamma_{\RZF}= \! \min\! \left(\frac{\rho_1}{d_1^\tau} \qh_{SR}^{\dag}\qD\qh_{SR},
    \frac{\kappa \rho_2
}{d_1^\tau d_2^\tau}\|\qh_{SR}\|^2\|\qh_{RD}\|^2\right)\nonumber\\
&=  \!\min\! \left(\!\frac{\rho_2}{d_1^\tau}\hat{\qh}_{SR}^{\dag}{\sf diag}\left(0, 1,\!\cdots\!,1\right)\hat{\qh}_{SR},
    \frac{\kappa \rho_1}{d_1^\tau d_2^\tau}
\|\qh_{SR}\|^2\|\qh_{RD}\|^2\!\right)\nonumber\\
&= \!\min\! \left(\frac{\rho_1}{d_1^\tau}\|\tilde{\qh}_{SR}\|^2,
    \frac{\kappa \rho_2}{d_1^\tau d_2^\tau}\|\qh_{SR}\|^2\|\qh_{RD}\|^2\right),
\end{align}
where $\hat{\qh}_{SR} = {\bf \Phi}{\qh}_{SR}$ with $\bf\Phi$ is a unitary matrix, and $\tilde{\qh}_{SR}$ is a $(M_R-1)\times 1$ vector, consisting of the $M_R-1$ last element of $\hat{\qh}_{SR}$. In~\eqref{eqn:SNR:R2}, the first equality follows
from the fact that $\qD$ is idempotent and the second equality is due to the eigen decomposition. Let us denote $Z_1 \triangleq\frac{1}{d_1^{\tau}}\left( \|\thSR\|^2  + |\tilde{h}_1|^2\right)$  where $\hat{h}_{1}$ is the first element of the $\hat{\qh}_{SR}$ and $X_1 = \frac{\|\thSR\|^2}{\|\thSR\|^2  + |\tilde{h}_1|^2}$. Therefore, the end-to-end SINR can be re-expressed as
\vspace{-0.4em}
\setcounter{equation}{47}
\begin{align}\label{eq:gammaRZF}
\gamma_{\RZF} = Z_1 \min  \left(\rho_1 X_1,\frac{\kappa\rho_2}{ d_2^{\tau}}
\|\hRD\|^2\right).
\end{align}
It is well known that $Z_1$ follows central chi-square distribution with $2M_R$ degrees-of-freedom, denoted as $Z_1\sim \chi^2_{2M_R}$ and that $X_1$ follows a beta distribution with shape parameters $M_R-1$ and $1$, denoted as $X_1 \sim \mathsf{Beta} (M_R-1, 1)$, with~\cite{MathematicalStatistics_1978}
\bea
F_{X_1}(x)\!=\!x^{M_R-1},\quad 0<x<1.
\eea
Moreover, let $Y_2 =\frac{\kappa\rho_2}{d_2^{\tau}}\|\hRD\|^2 $; we have $F_{Y_2}(y)=P\left(M_T, \frac{d_2^{\tau}}{\kappa\rho_2} y \right)$, where $P(a,x)=\gamma(a,x)/\Gamma(a)$. With $F_{X_1}(x)$ and $F_{Y_2}(y)$ in hand, $F_{\gamma_{\RZF}}(z)$ can be expressed as
\vspace{-0.4em}
\begin{align}\label{eq:cdf of gammaRZF}
F_{\gamma_{\RZF}}(z)
& \!= 1\!-\! Q\left(\!M_R, \bGsd\!\right)+
\\
&\frac{1}{\Gamma(M_R)}
\left( \int_{\bGsd}^{\infty}
P\left(\!M_T, \frac{d_1^{\tau}d_2^{\tau}}{\kappa \rho_2} \frac{z}{x} \!\right)
x^{M_R-1}e^{-x} dx
\right.\nonumber\\
&\left.
+ \left(\!\bGsd\!\right)^{M_R-1}
\int_{\bGsd}^{\infty}
Q\left(\!M_T, \frac{d_1^{\tau}d_2^{\tau}}{\kappa \rho_2} \frac{z}{x} \!\right)
e^{-x} dx \right).\nonumber
\end{align}

Although \eqref{eq:cdf of gammaRZF} does not admit a closed-form solution, it can be efficiently evaluated numerically. Now, we look into the high SNR regime, and investigate the diversity order.
\begin{proposition}\label{Propos:hsnr:RZF}
In the high SNR regime, i.e., $\rho_1, \rho_2  \rightarrow \infty$, the outage probability of the RZF scheme can be approximated as
\vspace{-0.4em}
\begin{align}\label{eq:cdf of gammaRZF: Asymptotic}
F_{\gamma_{\RZF}^{\infty}}(z)& \approx
 \begin{cases}
 \frac{1}{\Gamma(M_R)} \left(\bGsd\right)^{M_R-1},&\hspace{-7.5em} M_R < M_T+\!1,\\
 \frac{1}{\Gamma(M_R)}\left(1 \!+\! \frac{1}{\Gamma(M_T+1)} \left(\frac{\SnD}{\SnR}\frac{d_2^\tau}{\kappa}\right)^{M_T}\right)
 \left(\frac{d_1^{\tau}z}{\rho_1}\right)^{M_T}\!\!, & \\
 \hspace{12em} M_R = M_T+\!1,\\
 \frac{\Gamma(M_R-\!M_T)}{\Gamma(M_R)\Gamma(M_T+1)}\!\!
 \left(\frac{d_2^\tau}{\kappa}\right)^{M_T}\!\!
 \left(\frac{\Ssr z}{\rho_2}\!\right)^{M_T}, &\\
 \hspace{12em}M_R> M_T+\!1.
\end{cases}
\end{align}
\end{proposition}
\begin{proof}
See Appendix~\ref{proof:Apx:RZF:highSNR}.
\end{proof}

Proposition~\ref{Propos:hsnr:RZF} indicates that the RZF scheme achieves a diversity order of $\min(M_R-1, M_T)$. This result is also intuitively satisfying since one degree-of-freedom should be allocated for LI cancellation at the receive side of $R$.
\vspace{-0.3em}
\subsection{MRC/MRT Scheme}
The outage probability analysis of the MRC/MRT scheme for arbitrary $M_T$ and $M_R$ appears to be cumbersome. Therefore, we now consider two special cases as follows: Case-1) $M_T=1$, $M_R\geq 1$ and Case-2) $M_T \geq 1$, $M_R=1$.

Case-1): In this case $\frac{|\hSR^{\dag}\vH_{R R}\hRD^{\dag}|^2}{\|\hSR \hRD\|^2}$ is given by
\vspace{-0.3em}
\begin{align}\label{eq: WRHaaWT: MRC/MRT case-2}
X_2 &\triangleq |\qw_r^{\MRC} \vh_{R R}|^2=\frac{ \mid \hat{\vh}_{SR,1}\mid^2 }{ \|\hSR \|^2}  \| \vh_{RR}\|^2.
\end{align}
For notational convenience, we define  $c_1 = \frac{\rho_1}{ d_1^{\tau}}$,  $c_2 = \frac{\kappa \rho_1 \Sap}{ d_1^{\tau}}$, $c_3 =\frac{\kappa \rho_2}{ d_1^{\tau}d_2^{\tau}} $. Then, the end-to-end SINR can be re-expressed as
\vspace{-0.4em}
\begin{align}\label{eq:gammaMRC/MRT Case-1}
\gamma_{\MRC} &=\!
\min\left(\!   \frac{c_1 \|\hSR\|^2}
{c_2|\hat{\vh}_{SR,1}|^2  \| \vh_{RR}\|^2\! +\! 1},
c_3 \|\hSR\|^2 |h_{RD}|^2\!\right).
\end{align}
Let us denote $X  = c_1/\left(c_2 X_2  +\frac{1}{Y_3}\right)$ where $X2 =\frac{|\hat{\vh}_{SR,1}|^2  \| \vh_{RR}\|^2}{ \|\hSR\|^2} $ and $Y  = c_3Y_3Y_4$, with $Y_3=\|\hSR\|^2$ and $Y_4 = |h_{RD}|^2$.
Accordingly, the cdf of $\gamma^{\MRC}$ in~\eqref{eq:gammaMRC/MRT Case-1} can be expressed as
\vspace{-0.4em}
\begin{align}\label{eq:eq:OP exact case-1 joint cdf of X1 and X2}
F_{\gamma_{\MRC}}(z)&=\Prob( \min(X,Y) <z),\nonumber\\
 &= 1- \Prob(X>z, Y>z).
\end{align}
Conditioned on $Y_3$, the RVs $X$ and $Y$ are independent and hence we have
\vspace{-0.3em}
\begin{align}\label{eq:eq:OP exact case-1 joint cdf of X1 and X2}
 &\Prob(X>z, Y>z) \\\nonumber
 &=\int_{\frac{ \Ssr z}{\rho_1}}^{\infty} \left(1- F_{X |Y_3} (z))(1-F_{Y|Y_3} (z)\right) f_{Y_3}(y) dy,\\\nonumber
 &=\int_{\frac{ \Ssr z}{\rho_1}}^{\infty}\!
 F_{X_2} \left(\frac{1}{c_2}\left(\frac{c_1}{z}\! -\frac{1}{y}\right)\right)\!\left(1\!-F_{Y_4}\left(\frac{z}{c_3 y}\right)\right)f_{Y_3}(y) dy.
\end{align}
In order to evaluate~\eqref{eq:eq:OP exact case-1 joint cdf of X1 and X2} we require the cdf of the RV, $X_2$. Note that $ \|\vh_{RR}\|^2 \sim\!\chi_{2M_R}^2$, $ Z_2\triangleq\frac{ \mid \hat{\vh}_{SR,1}\mid^2 }{\|\hSR \|^2}$ is distributed as $Z_2\sim \Bta(1,M_R\!-\!1)$~\cite{MathematicalStatistics_1978}, and the cdf of $X_2$ can be readily evaluated as~\cite{Mohammadi:TCOM:2015}
 \vspace{-0.4em}
\begin{align}\label{eq:cdf of Y:chi2 times beta}
F_{X_2}(t)
=G_{2 3}^{2 1} \left(t \  \Big\vert \  {1, M_R \atop 1, M_R, 0} \right).
\end{align}
Now, using the cdf of RV $Y_3$, and substituting~\eqref{eq:cdf of Y:chi2 times beta} into~\eqref{eq:eq:OP exact case-1 joint cdf of X1 and X2} we obtain
\vspace{-0.4em}
\begin{align}\label{eq:OP exact case-1 final}
&F_{\gamma_{\MRC}}(z)
= 1-\frac{1}{\Gamma(M_R)}\\
 &
 \times\!\!\int_{\frac{ \Ssr z}{\rho_1}}^{\infty}\!\!
G_{2 3}^{2 1} \left( \frac{1}{c_2}\left(\frac{c_1}{z}\!- \!\frac{1}{y}\right)\ \!\! \Big\vert \ \!\!  {1, M_R \atop 1, M_R, 0} \right)
y^{M_R-1} e^{-\left(y+\frac{z}{c_3 y}\right)} dy.\nonumber
\end{align}
To the best of the authors' knowledge, the integral in~\eqref{eq:OP exact case-1 final} does not admit a closed-form
solution. However, \eqref{eq:OP exact case-1 final} can be evaluated numerically. We now look into the high SNR regime to gain more insights. To this end, neglecting the noise term of the first term inside in minimum function in~\eqref{eq:gammaMRC/MRT Case-1}, we write
\vspace{-0.4em}
\begin{align}\label{eq:gammaMRC/MRT Asymptotic Case-1}
\gamma_{\MRC}^{\mathsf{low}} > \gamma_{\MRC}  = \min\left(   \frac{1}{\kappa X_2} ,\frac{\kappa \rho_2}{ \Ssr\Srd} \parallel\hSR\parallel^2 |h_{RD}|^2\right).
\end{align}
We now present the following proposition.

\begin{proposition}\label{Propos:hsnr:MRC1}
In the high SNR regime, i.e., $\rho_1, \rho_2\rightarrow\infty$, with $M_T=1$ the outage probability of the MRC/MRT scheme can be approximated as
\vspace{-0.3em}
\begin{align}\label{eq:lower bound for cdf of gammaMRCMRT_Case1}
F_{\gamma_{\MRC}^{\mathsf{low}}}(z)
 &= 1-
\frac{2}{\Gamma(M_R)}
G_{2 3}^{2 1} \left(\frac{1}{\kappa \Sap z}  \  \Big\vert \  {1, M_R \atop 1, M_R, 0} \right)\nonumber\\
&\qquad\times\left(\frac{d_1^{\tau}d_2^{\tau}}{\rho_2 \kappa} z\right)^{\frac{M_R}{2}}\!\! K_{M_R}\left(2\sqrt{\frac{ d_1^{\tau}d_2^{\tau}}{\rho_2 \kappa} z}\right).
\end{align}
\end{proposition}
\begin{proof}
See Appendix~\ref{proof:Apx:MRC1:highSNR}.
\end{proof}
Moreover, by applying a Bessel function approximation for small arguments~\cite[Eq. (9.6.9)]{Abramowitz_Handbook_1970}, in~\eqref{eq:lower bound for cdf of gammaMRCMRT_Case1} we can write
\vspace{-0.3em}
\begin{align}\label{eq:error floor}
&F_{\gamma_{\MRC}^{\mathsf{low}}}(z)
 \rightarrow 1- G_{2 3}^{2 1} \left(\frac{1}{\kappa \Sap z}  \  \Big\vert \  {1, M_R \atop 1, M_R, 0} \right).
\end{align}
Note that~\eqref{eq:error floor} presents the outage probability floor and indicates that the MRC/MRT scheme with $M_T=1$ exhibits a zero-diversity order behavior in presence of residual LI.

Case-2) In this case $\frac{|\hSR^{\dag}\vH_{R R}\hRD^{\dag}|^2}{\|\hSR \hRD\|^2}$ simplifies to
\vspace{-0.3em}
\begin{align}\label{eq: WRHaaWT: MRC/MRT case-1}
Y_5 &\triangleq|\vh_{R R} \WT^{\MRT}|^2 = (\vh_{R R} \boldsymbol{\Phi}_t
\diag\{1,0,\cdots,0\} \boldsymbol{\Phi}^{\dag}_t\vh_{R R}^{\dag} )\nonumber\\
&= |\hat{\vh}_{RR,1}|^2,
\end{align}
where $\boldsymbol{\Phi}_t$ is an unitary matrix and follows from eigen decomposition and $\hat{\vh}_{RR} =\vh_{RR} \boldsymbol{\Phi}_t$. Hence, $\gamma_{\MRC}$ can be written as
\vspace{-0.3em}
\begin{align}
\gamma_{\MRC} &=
\min\left(   \frac{c_1 |h_{SR}|^2}
{c_2
|h_{SR}|^2 |\hat{\vh}_{RR,1}|^2 + 1}, c_3|h_{SR}|^2 \|\hRD\|^2\right).
\end{align}
Let us define $U = \frac{c_1 X_3}{c_2 X_3 Y_5 + 1}$ and $V = c_3X_3Y_6$,
where $X_3=|h_{SR}|^2$, $Y_6=\|\hRD\|^2$. Note that conditioned on $X_3$, the RVs, $U$ and $V$ are independent and hence we have
\vspace{-0.3em}
\begin{align}
\label{eq:cdf of gammaMRCMRT_Case2 exact}
&F_{\gamma_{\MRC}}(z)\\
&=1- \int_{\frac{ \Ssr z}{\rho_1}}^{\infty} \left(1- F_{U| X_3} (z)\!\right)\!\!\left(1\!-\!F_{V|X_3}(z)\right) f_{X_3}(x) dx,\nonumber\\
&=1-\int_{\frac{ \Ssr z}{\rho_1}}^{\infty}
 \left(1-e^{-\frac{1}{c_2x} \left(\frac{c_1 x}{z}-1\right)}\right)
 Q\left(M_T,\frac{z}{c_3 x}\right)
  e^{-x} dx.\nonumber
\end{align}
Having obtained the exact outage probability expression, we now look into the high SNR regime and establish the following asymptotic outage probability approximation.
\begin{proposition}\label{Propos:hsnr:MRC2}
In the high SNR regime, i.e., $\rho_1,\rho_2\rightarrow\infty$, with $M_R=1$, the outage probability of the MRC/MRT scheme can be approximated as
\vspace{-0.3em}
\begin{align}\label{eq:cdf of gammaMRCMRT_Case2 series approx}
&F_{\gamma_{\MRC}^{\infty}}(z)  \approx 1-
\left(1-e^{-\frac{1}{\kappa \Sap z} }\right)\nonumber\\
&\hspace{4em}\times
\left(e^{-\bGsd} \!-\! \frac{1}{\Gamma(M_T)}\frac{d_1^{\tau}z}{\rho_1}\sum_{k=0}^{\infty} \frac{(-1)^{k} }{k! (M_T+\!k)}\right.\nonumber\\
&\hspace{4em}\times\left. \left(\frac{\rho_1}{\rho_2}\frac{d_2^{\tau}}{\kappa} \!\right)^{M_T+k}\!\!\!\! E_{M_T+k}\left(\bGsd\right)\right)\!.
\end{align}
\end{proposition}
\begin{proof}
Note that in the high SNR regime, i.e., $\rho_1,\rho_2 \rightarrow \infty$, the cdf $F_{\gamma^{\MRC}}(z)$ can be approximated as
\vspace{-0.2em}
\begin{align}\label{eq:cdf of gammaMRCMRT_Case2 approx}
F_{\gamma_{\MRC}^{\infty}}(z) & \approx 1-
\frac{\left(1-e^{-\frac{c_1}{c_2z} }\right)}
{\Gamma(M_T)}
\int_{\frac{ \Ssr z}{\rho_1}}^{\infty}
 \Gamma\left(M_T,\frac{z}{c_3 x}\right)
  e^{-x} dx.
\end{align}
Now applying~\cite[Eq. (8.354.2)]{Integral:Series:Ryzhik:1992} and using the definition of the E$_{n}$-function we can easily obtain the desired result.
\end{proof}
In most cases, $k=0$ is sufficient for accurate results leading to the following compact expression
\vspace{-0.2em}
\begin{align}\label{eq:cdf of gammaMRCMRT_Case2 series approx compact}
&F_{\gamma_{\MRC}^{\infty}}(z)\approx 1-
\left(1-e^{-\frac{1}{\kappa\Sap}\frac{1}{z} }\right)
\left(e^{-\bGsd} - \frac{ 1}{\Gamma(M_T+1)}\right.\nonumber\\
&\hspace{5em}\times\left.\left(\frac{\rho_1}{\rho_2}\frac{\Srd}{\kappa} \right)^{M_T} \bGsd E_{M_T}\left(\bGsd\right)\right),
\end{align}
while in  other cases it can be truncated using few terms up to $10$.

At this point, it is important to determine the value of $\alpha$ that maximizes the throughput. We note that delay-constrained throughput should converge to the ceiling value of $R_c (1-\alpha)$ when $\Pout \rightarrow 0$. For each beamforming scheme we observe that $\Pout$ is a complicated function of $\alpha$ and it decreases as the value of $\alpha$ is increased. However, this will lead to the decrease of the term $(1-\alpha)$ at the same time. Therefore, an optimal value of $\alpha$ that maximizes the delay-constrained throughput exists and it can be found by solving the following optimization problem~\cite[Eq. (18)]{Caijun:TCOM:2014}
\vspace{-0.2em}
\begin{align}\label{eq:alpha:DC}
&\alpha^{*}  = \argmax_{0<\alpha<1} R (\alpha).
\end{align}
Given~\eqref{eq:cdf of gammaTZF},~\eqref{eq:cdf of gammaRZF},~\eqref{eq:OP exact case-1 final} and~\eqref{eq:cdf of gammaMRCMRT_Case2 exact}, unfortunately the optimization problem in~\eqref{eq:alpha:DC} does not admit closed-form solutions.
However, the optimal $\alpha^{*}$ can be solved numerically.
\vspace{-0.2em}
\section{Numerical Results and Discussion}\label{sec:numerical results}
We now present numerical results based on analytical expressions developed and investigate the impact of key system parameters on the performance.  The simulations adopt parameters of the 3GPP LTE for small cell deployments~\cite{LTE}. The maximum transmit power of the source node is set to $26$ dBm.  The energy conversion efﬁciency is set to be $\eta=0.5$.\footnote{We note that the typical values for practical parameters used in EH systems will depend on both the system application and specific technology used  for implementation of RF energy harvesting circuits.}
\subsection{Instantaneous Throughput}
We consider the influence of optimal time-split $\alpha$ and LI effect on the instantaneous throughput of different beamforming schemes. Fig.~\ref{fig: Instantanous-throu-v-alpha} shows the instantaneous throughput versus $\alpha$ of the beamforming schemes for a single time frame and channel realizations.
There are two groups of curves: the OPA (dashed line) and  EPA (solid line) curves. OPA and EPA refer to the cases where source allocates different and equal power levels for energy harvesting phase and the information transmission phase respectively. The results for OPA have been obtained using a two dimentional grid search stratergy over $\alpha$ and the splitting factor that controls the source power for the energy harvesting phase and the information transmission phase. It is clear that the OPA scheme achieves a higher throughput than the EPA scheme over the entire range of the time-split.
We can see that the values of the optimal $\alpha$ calculated by~\eqref{eq:opt alpha opt 3},~\eqref{eq:opt-alpha-TZF},~\eqref{eq:opt-alpha-RZF}, and~\eqref{eq:opt-alpha-MRC} coincide with the corresponding ones obtained via simulation. Also, as expected, the optimum scheme outperforms all other schemes on all time-split values. In addition, simulation results, not shown in the figure to avoid clutter, reveal that the values of the optimal $\alpha$ decrease as either the number of relays' receive antennas or the sources' transmit power increases. This is because in these cases the relay node can harvest the same amount of energy in a shorter time. Therefore, more time must be allocated to the information transmission phase in order to improve the system throughput.

\begin{figure}[t]
\vspace{-1em}
\centering
\includegraphics[width=85mm, height=70mm]{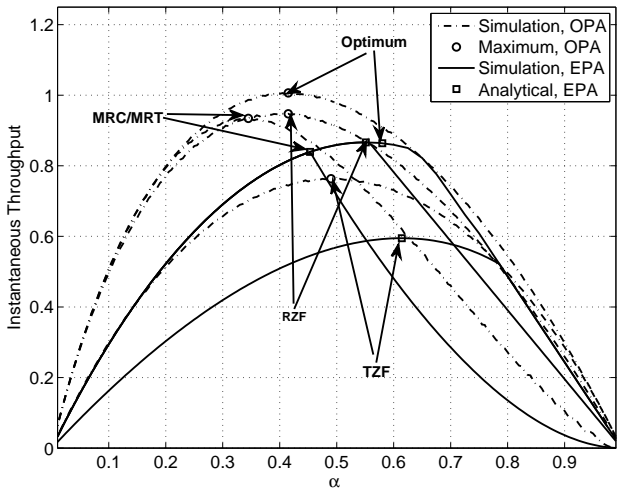}
\vspace{-1em}
\caption{Instantaneous throughput versus $\alpha$ for proposed beamforming schemes ($M_T=M_R=3$, $P_S=20$ dBm, $d_1=20$, $d_2=10$ and $\tau=3$). }
\label{fig: Instantanous-throu-v-alpha}
\vspace{-1em}
\end{figure}
\begin{figure}[t]
\centering
\includegraphics[width=85mm, height=66mm]{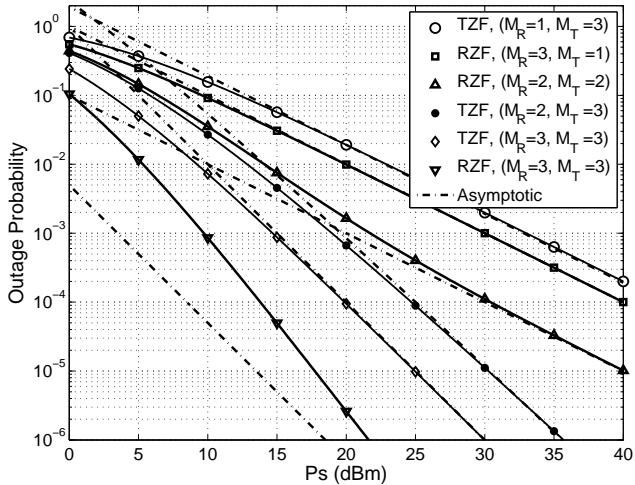}
\vspace{-0.9em}
\caption{Outage probability versus $P_S$ of the TZF and RZF schemes for different antenna configurations.}
\label{fig: Outage probability_ZF}
\vspace{-1.0em}
\end{figure}

\begin{figure}[t]
\centering
\vspace{0.2em}
\includegraphics[width=85mm, height=66mm]{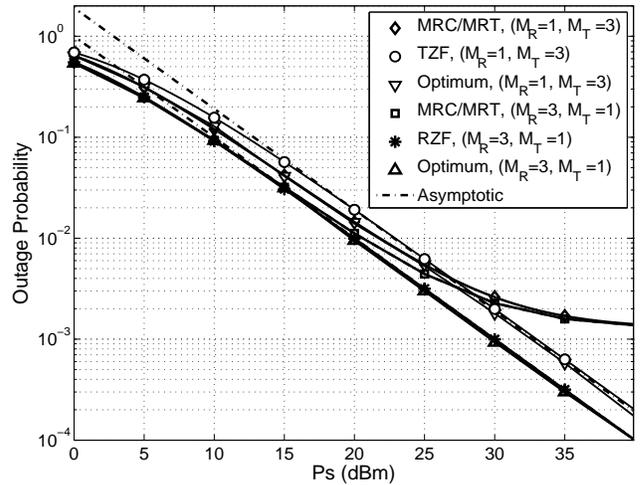}
\vspace{-0.9em}
\caption{Outage probability versus $P_S$ of the optimum, TZF, RZF, MRC/MRT schemes for different antenna configurations.}
\label{fig: Outage probability_MRC_ZF}
\vspace{-0.8em}
\end{figure}
\subsection{Outage Probability}
Fig.~\ref{fig: Outage probability_ZF} shows the outage probability for the ZF based beamforming schemes with different antenna configurations and for a specific $\alpha=0.5$ and $d_1=d_2 =10$ m. The asymptotic results based on~\eqref{eq:cdf of gammaTZF: Asymptotic} and~\eqref{eq:cdf of gammaRZF: Asymptotic} are also presented. Since the relay is capable of canceling LI, we see that the outage probability of the TZF and RZF beamforming schemes decays proportionally to the diversity orders reported in Proposition~\ref{Propos:hsnr:TZF} and~\ref{Propos:hsnr:RZF}, respectively.
Comparing the TZF and RZF schemes with the same diversity orders and different receive antenna numbers (i.e., TZF, with $M_R =1$ and $M_T=3$, and RZF  with $M_R =3$ and $M_T=1$) we see that the additional receive antenna could harvest more energy to increase the second-hop SNR and to facilitate information transfer. Moreover, for the case where $M_R=M_T$, RZF achieves a higher array gain. This observation demonstrates that while under some configurations ($M_T = 1$) or ($M_R = 1$) only one form (receive or transmit) of beamforming design can be implemented, when both designs can be applied, the system designer has to carefully decide on the configuration as well as the beamforming design.

Fig.~\ref{fig: Outage probability_MRC_ZF} compares the outage probability of the optimum, TZF, RZF, and MRC/MRT schemes with different antenna configurations and for $\alpha=0.5$ and $d_1=d_2 =10$ m. The residual LI strength at the relay is set to be $-50$ dBm. Asymptotic results in~\eqref{eq:lower bound for cdf of gammaMRCMRT_Case1} and~\eqref{eq:cdf of gammaMRCMRT_Case2 series approx compact} are also provided for the MRC/MRT scheme. The outage performance of the MRC/MRT scheme is almost identical to the optimum scheme at low SNRs, while the ZF-based schemes can achieve almost the same performance of the optimum scheme in the high SNR regime. When the transmit power increases and $\alpha$ remains fixed, an excessive amount of energy will be collected at the relay, which is detrimental for the MRC/MRT scheme since it results in a strong LI effect. Therefore, the outage probability of the MRC/MRT scheme shows an outage floor at high SNRs. However, to some level a proper choice of $\alpha$ can improve the outage performance of the MRC/MRT scheme. Fig.~\ref{fig: Outage probability_ZF} and~\ref{fig: Outage probability_MRC_ZF} illustrate that the outage probability of the system depends on the amount of energy harvested through the receive antenna numbers at the relay, source-relay link distance, and the energy conversion efficiency of the deployed energy harvester at the relay.  Moreover, it is significantly influenced by the transmit/receive beamforming design at the relay.

\begin{figure}[t]
\centering
\includegraphics[width=85mm, height=68mm]{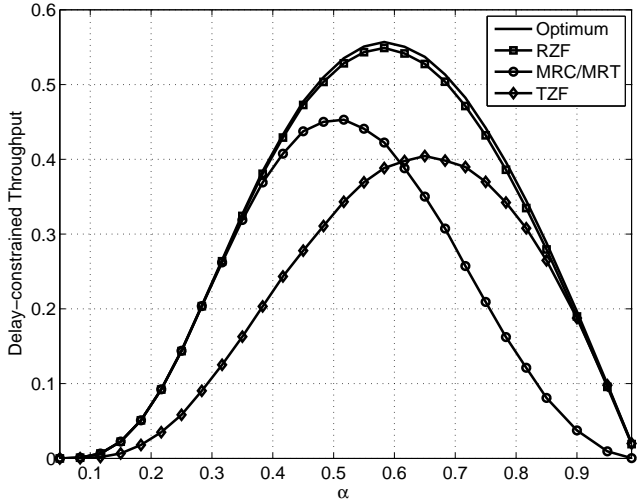}
\vspace{-0.7em}
\caption{Delay-constrained throughput of different FD schemes ($M_T=M_R=3$, $P_S=10$ dBm, $d_1=20$, $d_2=10$, $\tau=3$, $\eta=0.5$, and $R_c=2$).}
\label{fig: Delay-constraint-throu}
\vspace{-1.1em}
\end{figure}
\subsection{Delay-Constrained Throughput}
Fig.~\ref{fig: Delay-constraint-throu} shows the impact of optimal $\alpha$ on the delay-constrained throughput. As expected, the optimum scheme exhibits the best throughput among all beamforming schemes. The superior performance of the optimum scheme is more pronounced especially between $0.4$ and $0.8$ values of $\alpha$. The highest throughput with optimized $\alpha$ for the optimum, RZF, MRC/MRT and TZF schemes are given by $0.557$, $0.549$, $0.453$ and $0.404$, respectively. Moreover, we see that each one of the TZF, RZF and MRC/MRT schemes can surpass other beamforming schemes depending on the value of $\alpha$. This observation reveals the existence of various design choices when performance-complexity tradeoff is considered.

Fig.~\ref{fig: Delay-constraint-throu_v_LI} shows the effect of the LI strength on
the delay-constrained throughput when optimum $\alpha$ is used. The total number of antennas is $M_R+M_T=4$. The ZF-based schemes do not suffer from LI, therefore the delay-constrained transmission throughput remains constant. On the contrary, as expected, the delay-constrained transmission throughput of the MRC/MRT scheme decreases as $\Sap$ and consequently the LI strength increases. When the LI strength is low  the ZF-based schemes become inferior as compared to the MRC/MRT scheme. In this region the combination ($M_R=2$, $M_T=2$) with MRC/MRT scheme
exhibits a near optimum performance. Results, not shown for the configuration ($M_R=3$, $M_T=1$),  showed inferior performance as compared to the configuration ($M_R=2$, $M_T=2$) in all schemes while they showed a superior performance as compared to the configuration ($M_R=1$, $M_T=3$). Therefore, performance enhancements can be achieved through equal transmit and receive antenna deployment.

As a final cautionary note we would like to express that the above main findings and insights can be further examined by applying specific RF circuitry and power amplifier model. For example, low efficiency of a RF amplifier used for digital communications can be accounted in a detailed analysis. An implication of such use would be that, often the relay will be unable to operate and the price paid would be a high outage probability due to communication blackout periods. Nevertheless, our results provide useful theoretical performance bounds for the studied system and motivates practical interest from the perspective of wireless-powered full-duplex system implementation.
\vspace{-0.5em}
\section{Conclusion}\label{sec:conclusion}
In this paper, we studied the instantaneous and delay-constrained throughput of a wireless-powered FD MIMO relay system. We designed optimum linear processing at the relay as well as investigated several suboptimum schemes. The optimal time-split for instantaneous throughput maximization of all schemes was derived. We also presented exact and asymptotic closed-form expressions for the outage probability useful to characterize the delay-constrained throughput. We found that the MRC/MRT scheme can offer a higher instantaneous/delay-constrained throughput as compared to the RZF and TZF schemes, when the LI is significantly canceled, and vice versa. The MRC/MRT scheme can provide a better outage performance at low-to-medium SNRs, while the ZF precoders outperform the former at high SNRs.
\begin{figure}[t]
\centering
\vspace{-1.0em}
\includegraphics[width=85mm, height=71mm]{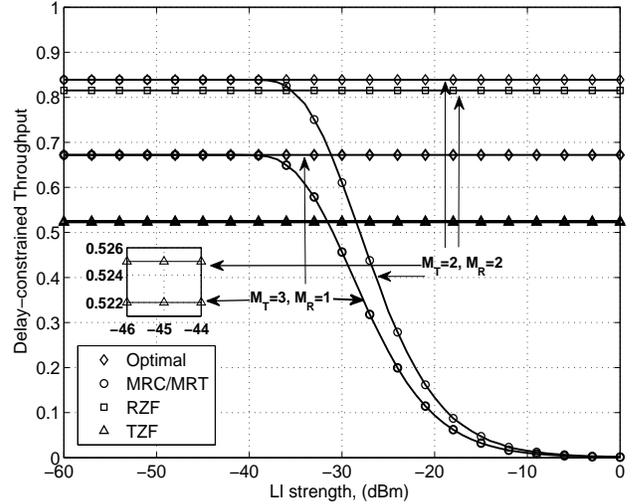}
\vspace{-0.6em}
\caption{Delay-constrained throughput versus the LI strength for proposed beamforming schemes with optimal $\alpha$ ($P_S=20$ dBm, $d_1=20$, and $d_2=10$).}
\label{fig: Delay-constraint-throu_v_LI}
\vspace{-1.0em}
\end{figure}
\appendices
\vspace{-1em}
\section{Proof of Proposition ~\ref{ProofPropositionRankOne}}
\label{Appendix:ProofPropositionRankOne}
\begin{proof}
We show that \eqref{eqn: nonconvex quadratic optimization 2} leads to optimum rank-one solution or such solution can be recovered from the optimum solution that has rank higher than one.
First note that instead of solving~\eqref{eq:opt-prec2} for a given $t$,~\eqref{eqn: nonconvex quadratic optimization 2} can be solved for a given $y={\bar y}$. Such optimization problem is expressed as
\vspace{-0.1em}
\bea\label{eq:KKTProof1}
\max_{ \qW_t, t} &&\hspace{-1em} t \nonumber\\
\mbox{s.t.} && \hspace{-1em}
\frac{\Ssr}{\rho_1}t{\bar y}\leq  \|\qh_{SR}\|^2  {\bar y} - \frac{\kappa \rho_1}{\Ssr}\|\qh_{SR}\|^2 \qh_{SR}^\dag\qH_{RR}\qW_t\qH_{RR}^\dag\qh_{SR} \nonumber \\
&&\hspace{-1em}
 {\bar y}= \mathsf{tr} \left(\qW_t \left(\qI+ \frac{\kappa \rho_1}{\Ssr}\|\qh_{SR}\|^2 \qH_{RR}^{\dag}\qH_{RR}\right)\right)\nonumber\\
&&\hspace{-1em}
 t\leq \frac{\kappa\rho_2 }{\Ssr\Srd} \|\qh_{SR}\|^2\qh_{RD} \qW_t \qh_{RD}^{\dag}\nonumber\\
&&\hspace{-1em}
\mathsf{tr} \left(\qW_t\right)=1, \qW_t\succeq 0
\eea
where ${\bar y}\in[1, y_{\rm up}] $. Here,  $y_{\rm up}$ is the maximum value of ${\bar y}$ which is the maximum eigenvalue corresponding to $ \qR\triangleq \qI+ \frac{\kappa \rho_1}{\Ssr}\|\qh_{SR}\|^2 \qH_{RR}^{\dag}\qH_{RR}$. The optimum  $\qW_t$ is the solution that maximizes $t$ for
all $y\in[1, y_{\rm up}] $ for which \eqref{eq:KKTProof1} is feasible. The optimum $\qW_t$ thus obtained will be the same as that obtained by solving~\eqref{eq:opt-prec2}. The Lagrangian multiplier function for~\eqref{eq:KKTProof1}  is expressed as
\vspace{-0.3em}
\begin{align}
\label{eq:KKTProof2}
{\mathcal L}\left(\qW_t, t, \left\{\lambda_i\right\}_{i=1}^{4}\right)&=-t+\lambda_1\left(\frac{\Ssr}{\rho_1}  t{\bar y}-\|\qh_{SR}\|^2  {\bar y}  \right.\nonumber\\
&\left.\hspace{-1em}+\frac{\kappa \rho_1}{\Ssr}\|\qh_{SR}\|^2 \qh_{SR}^\dag\qH_{RR}\qW_t\qH_{RR}^\dag\qh_{SR} \right)\nonumber\\
&  \hspace{-1em}+\lambda_2\left( {\rm tr}\left(\qW_t\qR\right)-{\bar y}\right)\nonumber\\
& \hspace{-1em}+ \lambda_3\left( t-\frac{\kappa \rho_2}{\Ssr\Srd} \|\qh_{SR}\|^2\qh_{RD} \qW_t \qh_{RD}^{\dag}  \right)\nonumber\\
& \hspace{-1em}+\lambda_4\left({\rm tr}(\qW_t\!-1)\right)\!-{\rm tr}\left( \qY \qW_t\right)\!,
\end{align}
where ${\bf Y}\succeq 0$ is the matrix dual variable associated with the constraint $\qW_t\succeq 0$. Since \eqref{eq:KKTProof1} is convex and Slater condition holds true, the Karush-Kuhn-Tucker (KKT) conditions are necessary and sufficient for optimality.
The KKT conditions for \eqref{eq:KKTProof1} are
\begin{eqnarray}
\label{eq:KKTProof3A}
\frac{\partial {\mathcal L}}{\partial \qW_t}&=0&\rightarrow {\bf Y}=\lambda_1  \frac{\kappa \rho_1}{\Ssr}\|\qh_{SR}\|^2 \qH_{RR}^\dag\qh_{SR} \qh_{SR}^\dag\qH_{RR}\nonumber\\
&&\hspace{-4em}+\lambda_2\qR+\lambda_4\qI-\lambda_3  \frac{\kappa \rho_2}{\Ssr\Srd} \|\qh_{SR}\|^2\qh_{RD}^{\dag} \qh_{RD} \\
\label{eq:KKTProof3B}
\frac{\partial {\mathcal L}}{\partial t}&=0&\rightarrow \lambda_1\frac{\Ssr}{\rho_1}{\bar y}+\lambda_3=1 \\
\label{eq:KKTProof3C}
& &\hspace{-4em} \lambda_1\left(\frac{\Ssr}{\rho_1}t{\bar y}-\|\qh_{SR}\|^2  {\bar y} +\frac{\kappa \rho_1}{\Ssr}\right.\nonumber\\
&&\hspace{-4em}\times\left\|\qh_{SR}\|^2 \qh_{SR}^\dag\qH_{RR}\qW_t\qH_{RR}^\dag\qh_{SR} \right)=0\\
\label{eq:KKTProof3D}
& &\hspace{-4em} \lambda_3\left( t\!-\!\frac{\kappa\rho_2 }{\Ssr\Srd} \|\qh_{SR}\|^2\qh_{RD} \qW_t \qh_{RD}^{\dag}  \right)\!=\!0\\
\label{eq:KKTProof3E}
&& \hspace{-4em}{\rm tr}\left(\qW_t\qR\right)-{\bar y}=0\\
\label{eq:KKTProof3F}
&&\hspace{-4em}{\rm tr}\left(\qW_t\right)-1=0\\
\label{eq:KKTProof3G}
& &\hspace{-4em} {\rm tr}(\qY \qW_t)=0 \rightarrow \qY \qW_t={\bf 0}.
\end{eqnarray}
The complementary slackness condition \eqref{eq:KKTProof3G} means that at KKT optimality, the optimum $\qW_t$ lies in the null-space of $\qY$. This means that the rank of   $\qW_t$ is the nullity of $\qY$.
In the following, we analyze the cases in which the optimum $\qW_t$ is rank-one (Cases a and c) and  not rank-one but the rank-one optimum solution can be recovered from optimum  $\qW_t$ (Case b).
\begin{itemize}
\item Case a: $\lambda_1=0, \lambda_3\neq 0$ - When $\lambda_1=0$, the  inequality constraint  corresponding to \eqref{eq:KKTProof3C} is not satisfied with equality, whereas the  inequality constraint corresponding to  \eqref{eq:KKTProof3D} is satisfied with equality since $\lambda_1=0$ leads to $\lambda_3=1$ (see  \eqref{eq:KKTProof3B}). In this case, $\qY$ reduces to
    \vspace{-0.3em}
    \begin{eqnarray}
\label{eq:KKTProof4}
\qY&=(\lambda_2+\lambda_4)\qI+\lambda_2  \frac{\kappa \rho_1}{\Ssr}\|\qh_{SR}\|^2 \qH_{RR}^{\dag}\qH_{RR}- \nonumber\\
&\frac{\kappa\rho_2 }{\Ssr\Srd} \|\qh_{SR}\|^2\qh_{RD}^{\dag} \qh_{RD}.
\end{eqnarray}
It is clear that   both $\lambda_2$ and $\lambda_4$ cannot be equal to zero at optimality. Otherwise,  $\qY$ turns to a negative semi-definite matrix contradicting the fact that $\qY\succeq 0$. In all other possible values of $\lambda_2 $ and $\lambda_4  $, it is seen that $\qZ\triangleq (\lambda_2+\lambda_4)\qI+\lambda_2  \frac{\kappa \rho_1}{\Ssr}\|\qh_{SR}\|^2 \qH_{RR}^{\dag}\qH_{RR}$ is a full-rank matrix. Now, we can show that the nullity of ${ \bf Y}$ cannot be greater than one by contradiction. Assume that $\left\{ \qu_{y,q}, q=1,2\right\} \in {\mathcal Ns  }({ \bf Y}) $ where  ${\mathcal Ns  }({ \bf Y})$ denotes null-space of ${ \bf Y}$ and let  $\qa\qa^{\dag}=\frac{\kappa\rho_2 }{\Ssr\Srd} \|\qh_{SR}\|^2\qh_{RD}^{\dag} \qh_{RD}$. Then,
\vspace{-0.5em}
\begin{align}
\label{eq:KKTProof5}
{ \bf Y}\qu_{y,q}&=\qZ\qu_{y,q} -\qa\qa^{\dag}\qu_{y,q}\nonumber\\
&\rightarrow \qu_{y,q}= \qZ^{ -1}\qa \qa^\dag\qu_{y,q}, \forall q,
\end{align}
which shows that $\qu_{y,q}$ is an eigenvector of $\qZ^{ -1}\qa \qa^\dag$ corresponding to eigenvalue $1$. Since ${ \rm rank }(\qZ^{ -1}\qa \qa^\dag)=1$, it turns out that $q$ cannot take a value greater than $1$. This shows that the dimension of null space of $\qY$ is 1, and therefore, the rank of $\qW_t$ is  one.

\item Case b:  $\lambda_1\neq 0, \lambda_3= 0$ - When optimum $\lambda_3=0$, the inequality constraint associated with  \eqref{eq:KKTProof3C} will be satisfied with equality, whereas that associated with \eqref{eq:KKTProof3D} will not be satisfied with equality.
Note that $\lambda_3=0$ leads to $\lambda_1=\frac{\rho_1}{\Ssr}\frac{1}{\bar y}$. In this case, $\qY$ reduces to
\vspace{-0.3em}
\bea
\label{eq:KKTProof6}
\qY&=\frac{\rho_1}{\Ssr}\frac{1}{\bar y}\frac{\kappa \rho_1}{\Ssr}\|\qh_{SR}\|^2 \qH_{RR}^\dag\qh_{SR} \qh_{SR}^\dag\qH_{RR}+ \nonumber\\
&(\lambda_2+\lambda_4)\qI+\lambda_2  \frac{\kappa \rho_1}{\Ssr}\|\qh_{SR}\|^2 \qH_{RR}^{\dag}\qH_{RR}.
\eea
Note that a feasible $\qY$ is  the one which has at least a nullity of 1, since the optimum $\qW_t$ lies in the null-space of $\qY$. Therefore, at optimality both $\lambda_2$  and $\lambda_4$ should be zero, otherwise $\qY$ in \eqref{eq:KKTProof6}  turns to a full-rank matrix which is not feasible.  Consequently,  $\qY$  reduces to $\qY= \frac{\rho_1}{\Ssr}\frac{1}{\bar y}\frac{\kappa \rho_1}{\Ssr}\|\qh_{SR}\|^2 \qH_{RR}^\dag\qh_{SR} \qh_{SR}^\dag\qH_{RR}$ which is a rank-one matrix. Therefore, the optimum $\qW_t$ may not be rank-one. However, we show that optimum rank-one matrix can be recovered from optimum $\qW_t$ without loss of optimality.

\hspace*{0.3cm} Suppose  the optimum $\qW_t$ has a rank $r$ where $\qW_t=\sum_{q=1}^{r}\sigma_q\qu_q\qu_q^\dag$. $\sigma_q$  and $\qu_q, (q=1,\cdots, r)$ are,  respectively, the eigenvalues and eigenvectors of the matrix $\qW_t$. Furthermore, due to the equality constraint  \eqref{eq:KKTProof3F}, $\sum_{q=1}^{r}\sigma_q=1$. Substituting eigenvalue decomposition of $\qW_t$  into the condition \eqref{eq:KKTProof3G},  we find that
\vspace{-0.3em}
\begin{align}
\label{eq:KKTProof7}
& \sum_{q=1}^{r}\sigma_q\qu_q^\dag {\bf Y} \qu_q=0\nonumber\\
&\rightarrow \qu_q^\dag \left[ \qH_{RR}^\dag\qh_{SR} \qh_{SR}^\dag\qH_{RR} \right] \qu_q=0, \forall q
\end{align}
where the last step is due to the fact that $\qY\succeq 0$. Moreover, following is due to (\ref{eq:KKTProof7})
\vspace{-0.3em}
\begin{align}
\label{eq:KKTProof8}
\frac{\Ssr}{\rho_1}t&= \|\qh_{SR}\|^2 -\frac{1}{\bar y}\frac{\kappa \rho_1}{\Ssr}\|\qh_{SR}\|^2 \qh_{SR}^\dag\qH_{RR}\qW_t\qH_{RR}^\dag\qh_{SR}\nonumber\\
&\rightarrow t= \frac{\rho_1}{\Ssr}\|\qh_{SR}\|^2.
\end{align}
On the other hand, using the eigenvalue decomposition of $\qW_t$, the inequality constraint associated with  \eqref{eq:KKTProof3D} and equality constraint  \eqref{eq:KKTProof3E}, respectively, yield
\vspace{-0.5em}
\begin{align}
\label{eq:KKTProof9}
& t< \frac{\kappa \rho_2 }{\Ssr\Srd} \|\qh_{SR}\|^2\sum_{q=1}^{r}\sigma_q \qu_q^\dag \qh_{RD}^{\dag} \qh_{RD} \qu_q\\
\label{eq:KKTProof10}
&1+ \frac{\kappa \rho_1}{\Srd} \|\qh_{SR}\|^2 \sum_{q=1}^{r}\sigma_q \qu_q^\dag \qH_{RR}^\dag \qH_{RR} \qu_q={\bar y}.
\end{align}
It  can be observed from \eqref{eq:KKTProof8} that the optimum value of $t$ does not depend on ${\bar y}$ in the underlying case. An arbitrary ${\bar y}$ where ${\bar y}>1$ (see  \eqref{eq:KKTProof10} ) remains optimum.  Now we can show that by choosing a particular $\qu_q$ from a set $\left\{ \qu_q\right\}_{q=1}^{r}$ and a specific value of $\sigma_q$,  (\ref{eq:KKTProof9}) is not violated. Consider that ${\hat q}=\max_{q}\qu_q^\dag \qh_{RD}^{\dag} \qh_{RD} \qu_q$. Then, the maximum of   (\ref{eq:KKTProof9}) is achieved by choosing $\qu_{\hat q}$ with $\sigma_{\hat q}=1$. This choice does not affect the optimum objective value which remains \eqref{eq:KKTProof8} . As such, the optimum rank-one matrix recovered from $\qW_t$ turns to  $\sigma_{\hat q}\qu_{\hat q}\qu_{\hat q}^\dag$.

 \item Case c:  $\lambda_1\neq 0, \lambda_3\neq 0$ - In this case, both the inequality constraints ( \eqref{eq:KKTProof3C} and  \eqref{eq:KKTProof3D}) will be satisfied with equality. As such, $\qY$ is given by (\ref{eq:KKTProof3A}). As long as at least one of  $\lambda_2$  and $\lambda_4$ is non-zero, $\lambda_1  \frac{\kappa \rho_1}{\Ssr}\|\qh_{SR}\|^2 \qH_{RR}^\dag\qh_{SR} \qh_{SR}^\dag\qH_{RR}+\lambda_2\qR+\lambda_4\qI$ is a full-rank matrix. On the other hand, $\lambda_3  \frac{\kappa \rho_2}{\Ssr\Srd} \|\qh_{SR}\|^2\qh_{RD}^{\dag} \qh_{RD}$ is a rank-one matrix. Consequently, as in Case a, the optimum $\qW_t$ can be shown to be a rank-one matrix. Therefore, it is sufficient to show that both   $\lambda_2$  and $\lambda_4$ cannot be zero at the optimality. Towards this end, we use the method of contradiction. Assume that  $\lambda_2=0$  and $\lambda_4=0$. Then, $\qY$ reduces to
  \bea
  \vspace{-0.2em}
\label{eq:KKTProof11}
\qY&=\lambda_1  \frac{\kappa \rho_1}{\Ssr}\|\qh_{SR}\|^2 \qH_{RR}^\dag\qh_{SR} \qh_{SR}^\dag\qH_{RR}\nonumber\\
&-\lambda_3  \frac{\kappa \rho_2 }{\Ssr\Srd} \|\qh_{SR}\|^2\qh_{RD}^{\dag} \qh_{RD}.
\eea
Note that $\qY$ should be positive-semidefinite. Since $\qY$ in \eqref{eq:KKTProof11}  is the difference between two rank-one matrices, it can be readily shown from Weyl's inequalities for eigenvalues of sum of Hermitian matrices that  $\qY$ cannot remain positive-semidefinite except
in the case with $\lambda_3  \frac{\kappa \rho_2}{\Ssr\Srd} \|\qh_{SR}\|^2\qh_{RD}^{\dag} \qh_{RD}=0$, i.e., $\lambda_3=0$ for non-zero $\qh_{SR}$ and $\qh_{RD}$. This contradicts with the assumption $\lambda_3 \neq 0$, which consequently contradicts the assumption that both $\lambda_2$ and $\lambda_4$ are zero. This completes the proof of the proposition.
\end{itemize}\vspace{-1.2em}
\end{proof}
\vspace{-0.4em}
\section{Proof of Proposition~\ref{Propos:hsnr:TZF}}
\label{propos:Apx:TZF:highSNR}
Applying the series expansion of $\gamma(a,x)$ and $\Gamma(a,x)$ ~\cite[Eq. (8.354.1) and (8.354.2)]{Integral:Series:Ryzhik:1992} we have
\vspace{-0.35em}
\begin{align}\label{eq:cdf of gammaTZF high SNR}
F_{\gamma_{\TZF}} (z)
  &=
  1\! -\! \frac{1}{\Gamma(M_R)}
  \int_{\frac{ d_1^{\tau}z }{\rho_1}}^{\infty}
  \left(1\! -\! \frac{1}{\Gamma(M_T\!-\!1)}
  \right.\nonumber\\
  &\hspace{-2em}\left.
  \times\sum_{k=0}^{\infty}
  \frac{(-1)^k}{k!(k\!+\!M_T)}\left(\frac{\Ssr\Srd}{\kappa\rho_2}\frac{z}{x}\right)^{M_T\!+k-1}\right) x^{M_R-1}e^{-x} dx,\!\nonumber\\
  &=1\!-\!\frac{\Gamma\left(M_R, \frac{ d_1^{\tau}z }{\rho_1}\right)}{\Gamma(M_R)} +
  \frac{1}{\Gamma(M_T-1)\Gamma(M_R)}
  \nonumber\\
  &\hspace{-2em}
  \times
  \sum_{k=0}^{\infty}\frac{(-1)^k}{k!(k+M_T)}
   \left(\frac{\Ssr\Srd}{\kappa} \frac{z}{\rho_2}\right)^{M_T+k-1} \mathcal{I}(k),
\end{align}
where $\mathcal{I}(k)=\int_{\frac{ d_1^{\tau}z }{\rho_1}}^{\infty} x^{M_R - M_T-k} e^{-x} dx$, that has closed form solution given by
\vspace{-0.35em}
\begin{align}\label{eq:integral high SNR TZF}
\mathcal{I}(k)
& =
 \begin{cases}
 \!\!\left(\bGsd\right)^{M_R - M_T-k+1}\!\!E_{M_T - M_R+k} \left(\bGsd\right),  &\\
 \hspace{8em}
  M_T > M_R-k\\
 \!\!\left(\bGsd\right)^{M_R - M_T-k+1}\!\!\alpha_{M_R - M_T-k} \left(\bGsd\right),  &\\
  \hspace{8em}
  M_T \leq M_R-k
 \end{cases}
\end{align}
where $\alpha_{n}(x) =\int_{1}^{\infty} x^{n} e^{-x} dx$~\cite[Eq. (5.1.5)]{Abramowitz_Handbook_1970}. By substituting~\eqref{eq:integral high SNR TZF} into~\eqref{eq:cdf of gammaTZF high SNR}, and then applying~\cite[Eq. (5.1.8)]{Abramowitz_Handbook_1970} and~\cite[Eq. (5.1.12)]{Abramowitz_Handbook_1970}, we get~\eqref{eq:integral high SNR TZF final} at the top of the next page where $\psi(1)=-0.57721...$ and $\psi(n) = \psi(1) + \sum_{m=1}^{n-1} \frac{1}{m},$ for $n>1$~\cite[Eq. (9.73)]{Integral:Series:Ryzhik:1992}. In the high SNR regime, i.e., $\rho_1, \rho_2 \rightarrow \infty$, omitting higher order items of
the series expansion in~\eqref{eq:integral high SNR TZF final}, the desired result follows
after some simple algebraic manipulations.
\bigformulatop{86}{ \vskip-0.5cm
\begin{align}\label{eq:integral high SNR TZF final}
F_{\gamma_{\TZF}} (z)
&\!=\!
1\!-\!\frac{\Gamma\left(M_R, \frac{ d_1^{\tau}z }{\rho_1}\right)}{\Gamma(M_R)} \!+\!
  \frac{1}{\Gamma(M_T\!-\!1)\Gamma(M_R)}
  \left(\bGsd\right)^{M_R}\!\sum_{k=0}^{\infty}\frac{(-1)^k}{k!(k+M_T)}
  \left(\frac{\SnD}{\SnR}\frac{\Srd}{\kappa}  \right)^{M_T+k-1}
 \nonumber\\
 &
 \times\begin{cases}
 \Big(\frac{(-1)^{M_T \!-\! M_R+k-1}}{\Gamma(M_T -\! M_R +k)}
 \left(\bGsd\right)^{M_T \!-\! M_R +k -1}
 \!\!\left(-\ln\left(\bGsd\right)\!+\! \psi(M_T\!-\!M_R\!+\!k)\right)\\
 \quad
 -\sum\limits_{\substack{\ell=0 \\
       \ell\neq M_T - M_R +k-1}}^{\infty}
 \frac{(-1)^{\ell}}{\ell +M_R - M_T-k +1}
 \left(\bGsd\right)^{\ell}\Big),  & M_T > M_R-k\\
 \Gamma(M_R - M_T -k+ 1)\left(\bGsd\right)^{M_T \!-\! M_R-k-1}\\
 \quad
 \times e^{-\bGsd}
 \left( 1 \!+\! \bGsd \!+\! \frac{1}{2!}\left(\bGsd\right)^2 \!+\! \cdots+\frac{1}{(M_R\!-M_T-\!k)!}\left(\bGsd\right)^{M_R-\!M_T-\!k}\right),  & M_T \!\leq\! M_R\!-\!k
 \end{cases}
\end{align}
}
\section{Proof of Proposition~\ref{Propos:hsnr:RZF}}
\label{proof:Apx:RZF:highSNR}
\hspace*{-0.3cm} Using the series expansion of $\gamma(a,x)$ and $\Gamma(a,x)$, we get
\setcounter{equation}{87}
\begin{align}\label{eq:cdf of gammaRZF high SNR}
F_{\gamma_{\RZF}}(z)
& =
1 -
\frac{\Gamma\left(M_R, \bGsd\right)}{\Gamma(M_R)}+
\frac{1}{\Gamma(M_R)\Gamma(M_T)}
\nonumber\\
  &\hspace{-1em}
  \sum_{k=0}^{\infty}
\frac{(-1)^k}{k!(k+M_T)}
   \left(\frac{\Ssr\Srd}{\kappa} \frac{z}{\rho_2} \right)^{M_T+k}
   \mathcal{I}_1(k)+
   \nonumber\\
  &\hspace{-1em}
\frac{1}{\Gamma(M_R)}\left(\bGsd\right)^{M_R-1}
\left( e^{-\bGsd}  -
\frac{1}{\Gamma(M_T)}
\right.\nonumber\\
  &\hspace{-1em}\left.
\sum_{k=0}^{\infty}\frac{(-1)^k}{k!(k+M_T)}
   \left(\frac{\Ssr\Srd}{\kappa} \frac{z}{\rho_2} \right)^{M_T+k}
   \mathcal{I}_2(k)\right),
\end{align}
where $\mathcal{I}_1(k) =\int_{\bGsd}^{\infty} x^{M_R-M_T-k-1}e^{-x}dx$ is evaluated as
\begin{align}\label{eq:integral I1}
\mathcal{I}_1(k) =
 \begin{cases}
 \!\!\left(\bGsd\right)^{M_R - M_T-k}\!\!E_{M_T - M_R+k+1} \left(\bGsd\right)\!,  & \\
 \hspace{7em}M_T > M_R-k-1 \\
 \!\!\left(\bGsd\right)^{M_R - M_T-k}\!\!\alpha_{M_R - M_T-k-1} \left(\bGsd\right)\!,  &\\
 \hspace{7em} M_T \leq M_R-k-1
 \end{cases}
\end{align}
and $\mathcal{I}_2(k) =\int_{\bGsd}^{\infty} x^{-M_T-k}e^{-x}dx$ is solved as
\begin{align}\label{eq:integral I2}
\mathcal{I}_2 (k)=
 \left(\bGsd\right)^{-M_T -k +1}E_{M_T +k} \left(\bGsd\right).
\end{align}

Substituting~\eqref{eq:integral I1} and~\eqref{eq:integral I2} into~\eqref{eq:cdf of gammaRZF high SNR}, and expanding the result with the help of \cite[Eq. (5.1.8) and Eq. (5.1.12)]{Abramowitz_Handbook_1970}, and then only selecting the sufficient and ignoring the higher order terms, we arrive at~\eqref{eq:cdf of gammaRZF: Asymptotic}.

\section{Proof of Proposition~\ref{Propos:hsnr:MRC1}}
\label{proof:Apx:MRC1:highSNR}
Let $Z_3=\parallel\hSR\parallel^2 |h_{RD}|^2$, and observing that $X_2$ and $Z_3$ are independent, we have
\begin{align}\label{eq:cdf of gammaMRCMRT_first step}
F_{\gamma_{\MRC}^{\mathsf{low}}}(z)\! =\! 1\! -\! F_{X_2}\left(\!\frac{1}{\kappa \Sap }\frac{1}{z}\right)\!
\left(1 \!-\! F_{Z_3}\left(\!\frac{\Ssr\Srd}{\kappa \rho_2 } z\!\right)\!\right),
\end{align}
where $F_{X_2}(\cdot)$ is given in~\eqref{eq:cdf of Y:chi2 times beta} and the cdf of $Z_3$ can be derived as
\begin{align}\label{eq:cdf for the product of two chi-squared}
F_{Z_3}(z)  &= \int_{0}^{\infty} \Prob\left(|h_{RD}|^2<\frac{z}{x}\right)f_{\parallel\hSR\parallel^2}(x) dx,\nonumber\\
&=1- \frac{2 }{\Gamma(M_R)} z^{\frac{M_R}{2}} K_{M_R}(2\sqrt{z}),
\end{align}
where we used~\cite[Eq. (3.471.9)]{Integral:Series:Ryzhik:1992} to derive~\eqref{eq:cdf for the product of two chi-squared}.
Finally, substituting~\eqref{eq:cdf of Y:chi2 times beta} and~\eqref{eq:cdf for the product of two chi-squared} into~\eqref{eq:cdf of gammaMRCMRT_first step} yields the desired result.

\balance
\bibliographystyle{IEEEtran}


\end{document}